\documentclass[a4paper,fleqn]{cas-dc}

\usepackage[numbers,compress]{natbib}

\def\tsc#1{\csdef{#1}{\textsc{\lowercase{#1}}\xspace}}
\tsc{WGM}
\tsc{QE}
\tsc{EP}
\tsc{PMS}
\tsc{BEC}
\tsc{DE}

\usepackage{amsmath,amsfonts,amssymb,amsthm}
\usepackage{mathtools}
\usepackage{bm}
\usepackage{algorithm}
\usepackage{algorithmicx}
\usepackage{algpseudocode}

\usepackage{array}
\usepackage{textcomp}
\usepackage{stfloats}
\usepackage{placeins}
\usepackage{cuted}
\usepackage{url}
\usepackage{verbatim}
\usepackage{graphicx}

\usepackage{multirow}
\usepackage{booktabs}
\usepackage[utf8]{inputenc}
\usepackage{xcolor}
\usepackage{tikz}
\usetikzlibrary{arrows.meta, positioning, shapes.geometric}
\usepackage{caption}
\usepackage{subcaption}
\makeatletter
\AtBeginDocument{%
  \renewcommand{\figurename}{Fig.}%
  \def\fnum@figure{\textbf{\figurename\ \thefigure.}\@gobble}%
}
\makeatother
\usepackage{pifont}
\usepackage{enumitem}
\newtheorem{definition}{Definition}
\newtheorem{assumption}{Assumption}

\newtheorem{theorem}{Theorem}

\newtheorem{proposition}[theorem]{Proposition}
\newtheorem{remark}{Remark}

\begin{document}
\let\WriteBookmarks\relax
\setcounter{topnumber}{4}
\setcounter{bottomnumber}{2}
\setcounter{totalnumber}{6}
\setcounter{dbltopnumber}{2}
\renewcommand{\topfraction}{.95}
\renewcommand{\bottomfraction}{.90}
\renewcommand{\textfraction}{.05}
\renewcommand{\floatpagefraction}{.80}
\renewcommand{\dbltopfraction}{.95}
\renewcommand{\dblfloatpagefraction}{.80}
\shorttitle{CILER for out-of-distribution recommendation}
\shortauthors{Wang et~al.}

\title[mode=title]{Conditionally Identifiable Latent-Environment Modeling for Out-of-Distribution Recommendation}

\author[1,2]{Qianqian Wang}
\ead{12331103@mail.sustech.edu.cn}

\author[1]{Wenwu Gong}
\cormark[1]
\ead{gongww@sustech.edu.cn}

\author[1]{Yunshan Li}
\ead{lyunshan565@gmail.com}

\author[1]{Zhenqing Wu}
\ead{12131251@mail.sustech.edu.cn}

\author[3]{Ruili Wang}\cormark[1]
\ead{ruili.wang@massey.ac.nz}

\author[1,2]{Lili Yang}\cormark[1]

\ead{yangll@sustech.edu.cn}

\affiliation[1]{organization={Shenzhen Key Laboratory of Safety and Security for Next Generation of Industrial Internet, Southern University of Science and Technology}, 
city={Shenzhen},
country={China}}
\affiliation[2]{organization={Department of Statistics and Data Science, Southern University of Science and Technology},
  city={Shenzhen},
  country={China}}

\affiliation[3]{organization={School of Mathematical and Computational Sciences, Massey University},
  city={Auckland},
  country={New Zealand}}

\cortext[1]{Corresponding authors}

\begin{abstract}
Out-of-distribution (OOD) recommendation is vulnerable to preference shifts induced by a latent environment. Existing methods can infer latent states from logged interactions, yet the statistical meaning of the latent environment and its effect on preference remain underdetermined. We formulate this task as conditionally identifiable risk-aware recommendation (CI-RR) and propose Conditionally Identifiable Latent-Environment Recommendation (CILER). CILER uses a user-conditioned exponential family to model the latent environment and a feature-indexed polynomial to specify how it changes preference. It predicts by marginalizing item probabilities over the inferred environment distribution. Under sufficient variation, correct specification, and decoder regularity, CILER identifies the environment-sensitive representation up to the stated equivalence class. We further bound excess deployment log-risk by environment-inference error. Controlled studies test the observable consequences of sufficient variation and model specification. Experiments on three datasets show that CILER improves all twelve OOD ranking metrics under feature, temporal, and geographical shifts within shared support.
\end{abstract}


\maketitle

\section{Introduction}
\label{sec:intro}

Out-of-distribution (OOD) recommendation considers recommender systems deployed under interaction distributions that differ from those observed during training\cite{wang2022causal}. A change in the latent environment can alter preference over item categories even when the user and candidate items remain the same, a pattern illustrated by the temporal shift on Meituan in Fig.~\ref{shifts}~\citep{Gan2024AttentionbasedCR,wang2024distributionally}. In such settings, a naive strategy is to learn stable preference by enforcing invariance across training environments~\citep{wang2022invariant,zhang2023invariant, li2023invariant}. This strategy reduces sensitivity to observed environment shifts, and it also discards environment-sensitive preference that remains relevant to deployment ranking. Generalizing in OOD recommendation requires separating the two kinds of preference.

Recent research infers latent environments from logged interactions and separates stable from environment-sensitive preference without environment labels~\citep{wang2023causal,zhang2024disentangled,zhao2025graph}. Reconstruction and disentanglement objectives can fit the observed interactions, yet they do not fix the statistical meaning of the latent environment. Observationally equivalent parameterizations can assign the same interaction variation differently between the stable and environment-sensitive components. They can also imply different environment distributions and preference mechanisms~\citep{locatello2019challenging,khemakhem2020variational,liu2024identifiable}. Additional sparsity or disentanglement penalties narrow the solution empirically, but they do not resolve this ambiguity. This limitation raises a pivotal issue:
\begin{center}
\begin{minipage}{0.94\linewidth}
\emph{How can we constrain a latent-environment recommender so that its environment-sensitive representation is identifiable?}
\end{minipage}
\end{center}

\begin{figure}
\includegraphics[width=1.0\columnwidth]{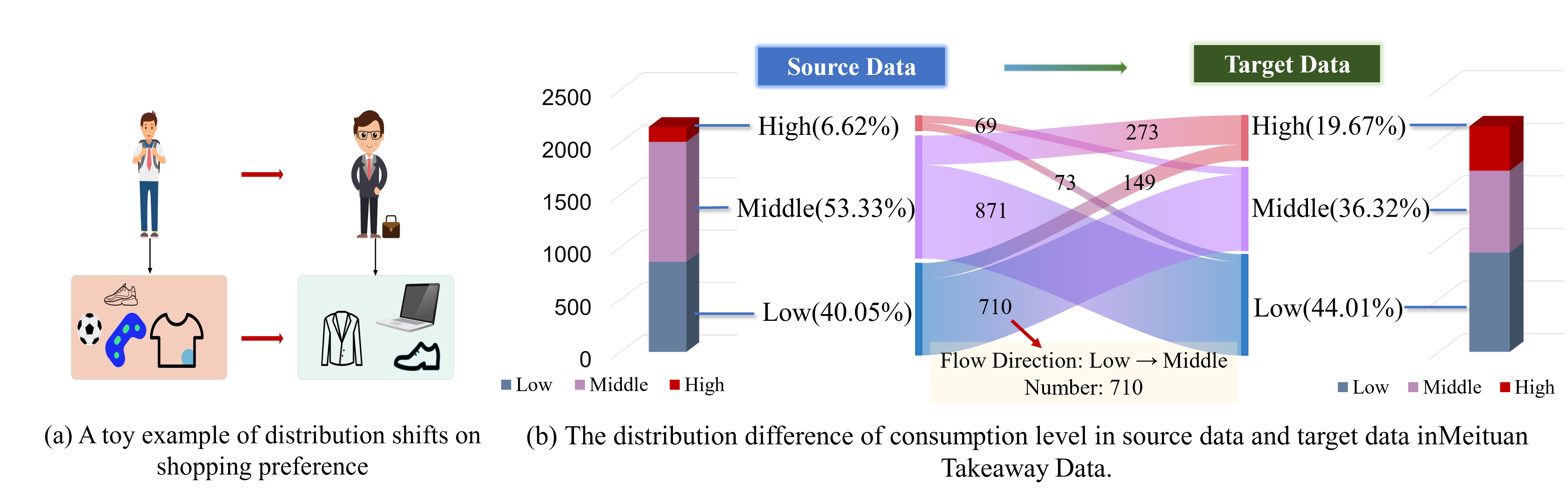}
\caption{Distribution shift in recommendation. A change in user context shifts preference over item categories and produces different interaction distributions in the Meituan training and deployment splits.}
\label{shifts}
\end{figure}

In this work, we formulate OOD recommendation as conditionally identifiable risk-aware recommendation (CI-RR). CI-RR seeks a predictor that minimizes deployment risk under an identifiability constraint (see Definition~\ref{def:cirr}). The constraint restricts observational equivalence to permutations and component-wise invertible transformations of the environment-sensitive representation. To address CI-RR, we propose Conditionally Identifiable Latent-Environment Recommendation (CILER). Logged interactions leave the latent-environment distribution and the preference mechanism undetermined, and CILER constrains both. A user-conditioned exponential family ties the latent-environment distribution to observed user context. A feature-indexed polynomial specifies the effect of the environment on preference. Together, the two restrictions define a model class that meets the CI-RR constraint under the conditions of Theorem~\ref{thm:ident}. At deployment, CILER marginalizes item probabilities over the inferred latent-environment distribution, so ranking reflects the whole environment distribution. Overall, CILER uses only logged interactions and observed user context, and it needs no environment labels or test-time updates. 

Theoretically, CILER meets the CI-RR identifiability constraint under sufficient variation, correct specification, and decoder regularity (Theorem~\ref{thm:ident}). Support overlap carries the recovered model to the deployment distribution and does not enter identifiability. The proof specializes the identifiable polynomial causal model of~\citet{liu2024identifiable} to the CILER construction. We further bound the excess deployment log-risk of environment marginalization by environment-inference error (Proposition~\ref{prop:risk}). KL contraction under marginalization makes this connection explicit. A smaller conditional KL error yields a tighter deployment-risk bound.

To evaluate CILER, we compare it with thirteen baselines on three datasets covering feature, temporal, and geographical shifts within shared support. CILER ranks first on all twelve OOD ranking metrics, with gains of up to 25.6\% over the strongest baseline. Controlled studies test the observable consequences of sufficient variation and model specification. Capacity-matched ablations isolate the two restrictions, and paired predictive diagnostics evaluate environment marginalization. Together, the experiments connect OOD ranking gains to the mechanisms and conditions stated in CI-RR.

Our contributions are summarized as follows.
\begin{itemize}
\item \textbf{Conditionally identifiable risk-aware formulation.} We formulate OOD recommendation as CI-RR (Definition~\ref{def:cirr}), which requires a predictor to minimize deployment risk under conditional identifiability of the environment-sensitive representation. The stable--sensitive split becomes an explicit requirement rather than an outcome of inductive bias.
\item \textbf{Constrained latent-environment modeling with marginalized prediction.} We propose CILER as an instantiation of CI-RR. A user-conditioned exponential family links the latent-environment distribution to observed user context, and a feature-indexed polynomial restricts the effect of the environment on preference. Prediction marginalizes over the inferred environment distribution and uses no environment labels and no test-time updates.
\item \textbf{Identifiability and deployment-risk guarantees.} We establish conditional identifiability for the CILER model class by specializing the result of~\citet{liu2024identifiable} to the proposed construction (Theorem~\ref{thm:ident}). We also bound excess deployment log-risk by environment-inference error (Proposition~\ref{prop:risk}). The two results connect environment inference to deployment prediction.

\end{itemize}

\section{Problem formulation}
\label{sec:prelim}
We first define latent-environment recommendation, then show what logged interactions identify and what they leave open. The two steps lead to the CI-RR problem, which makes identifiability an explicit constraint on the model class.

\subsection{Preliminaries}
Let $\mathcal{U}=\{u_1,\ldots,u_N\}$ and $\mathcal{I}=\{i_1,\ldots,i_M\}$ denote the user and item sets, and let $\boldsymbol{y}\in\{0,1\}^{N\times M}$ collect the observed interactions. For user $u$, the vector $\boldsymbol{x}_{1,u}$ contains the observed user context, $\boldsymbol{y}^{\mathrm{hist}}_u$ contains the available interaction history, $\boldsymbol{x}_{2,u}$ denotes unobserved user variation, and $\boldsymbol{\varepsilon}_u\in\mathcal{E}\subseteq\mathbb{R}^{k}$ denotes the latent environment. We write lowercase symbols for fixed values and drop the user index when no ambiguity arises. A predictor $f_{\Theta}(\boldsymbol{y}^{\mathrm{hist}}_u,\boldsymbol{x}_{1,u})\in\mathbb{R}^{M}$ scores items and incurs a ranking loss $\ell$. The user preference representation is $\boldsymbol{z}=[\boldsymbol{z}_1,\boldsymbol{z}_2]$, where the \emph{stable} component $\boldsymbol{z}_1$ receives no direct input from $\boldsymbol{\varepsilon}$ and the \emph{environment-sensitive} component $\boldsymbol{z}_2$ varies with it.

The observed context $\boldsymbol{x}_1$ and the interaction history are available both during training and at deployment. The latent variable $\boldsymbol{x}_2$ carries residual user heterogeneity, and $\boldsymbol{\varepsilon}$ carries context-dependent environmental variation. Neither of them is ever observed and recommender models infer them from the history and the observed context.

\noindent\textbf{Latent-environment recommendation.}
A latent-environment recommender comprises a conditional environment distribution $p(\boldsymbol{\varepsilon}\mid\boldsymbol{x}_1)$, a preference mechanism $\boldsymbol{z}_2=g(\boldsymbol{x}_2,\boldsymbol{\varepsilon})$, and a decoder $h$ from $\boldsymbol{z}$ to item scores. We write this model as $\mathcal{M}=(p(\boldsymbol{\varepsilon}\mid\boldsymbol{x}_1),g,h)$. The interaction likelihood alone does not pin down this triple. Section~\ref{sec:prelim2} makes the point precise by constructing observationally equivalent models whose stable--sensitive decompositions disagree.

We model deployment change as covariate shift in the observed features. The conditional environment and the preference mechanism stay fixed.

\begin{assumption}[Feature-mediated shift]
\label{as:shift}
Training and deployment differ only in the marginal distribution of observed user features. We assume $P_{\mathrm{tr}}(\boldsymbol{x}_1)\neq P_{\mathrm{te}}(\boldsymbol{x}_1)$, $P_{\mathrm{te}}\ll P_{\mathrm{tr}}$, and $dP_{\mathrm{te}}/dP_{\mathrm{tr}}\le B<\infty$. The conditional environment $p(\boldsymbol{\varepsilon}\mid\boldsymbol{x}_1)$, the user-variation law $p(\boldsymbol{x}_2)$, the preference mechanism $g$, and the decoder $h$ remain invariant.
\end{assumption}

Under Assumption~\ref{as:shift}, the deployment risk of a predictor is
\begin{equation}
\label{eq:robust_def}
\mathcal{R}_{\mathrm{te}}(f_{\Theta})=
\mathbb{E}_{\boldsymbol{x}_1\sim P_{\mathrm{te}}}\,
\mathbb{E}_{\boldsymbol{\varepsilon}\sim p(\cdot\mid\boldsymbol{x}_1)}\,
\mathbb{E}_{(\boldsymbol{y},\boldsymbol{x}_2)\sim P_{\boldsymbol{\varepsilon}}}
\big[\ell(f_{\Theta};\boldsymbol{y},\boldsymbol{x}_1)\big].
\end{equation}
Shifting the law of $\boldsymbol{x}_1$ moves probability mass toward regions that are rare during training yet still covered by it. Two ingredients control the resulting risk. Conditional identifiability fixes which coordinates of the representation react to the environment, and support overlap allows the fitted model to be evaluated under the shifted feature law. Feature values outside the training support lie beyond the scope of the analysis.

\subsection{Observational ambiguity and constrained prediction}
\label{sec:prelim2}
Fitting a recommender maximizes the likelihood of the observed interactions, so two models that induce the same interaction distribution stay indistinguishable even in the infinite-data limit.

\begin{definition}[Observational equivalence]
\label{def:oec}
Two latent-environment recommenders $\mathcal{M}$ and $\widetilde{\mathcal{M}}$ are observationally equivalent, written $\mathcal{M}\sim\widetilde{\mathcal{M}}$, if $p_{\mathcal{M}}(\boldsymbol{y}\mid\boldsymbol{x}_1)=p_{\widetilde{\mathcal{M}}}(\boldsymbol{y}\mid\boldsymbol{x}_1)$ for all $\boldsymbol{x}_1$ and all $\boldsymbol{y}$. The observational equivalence class of $\mathcal{M}$ within a model class $\mathcal{F}$ is
$\mathbb{E}_{\mathcal{F}}(\mathcal{M})=\{\widetilde{\mathcal{M}}\in\mathcal{F}:\widetilde{\mathcal{M}}\sim\mathcal{M}\}$.
\end{definition}

In the unrestricted class $\mathcal{F}_{\mathrm{all}}$, the conditional environment is an arbitrary smooth density and $g$ and $h$ are arbitrary smooth maps. The split of $\boldsymbol{z}$ into a stable block and an environment-sensitive block is part of the model, so the split itself must be identified. Two members of the class place the same coordinate on opposite sides of the split and still induce the same interaction distribution.

\begin{proposition}[Preference entanglement]
\label{prop:entangle}
Let $\mathcal{M}=(p(\boldsymbol{\varepsilon}\mid\boldsymbol{x}_1),g,h)\in\mathcal{F}_{\mathrm{all}}$. Then the following hold.
\begin{enumerate}[label=(\roman*),leftmargin=*]
\item For every diffeomorphism $\phi$ of $\mathcal{E}$, the model
$\mathcal{M}_{\phi}=\bigl(\phi_{\#}p(\cdot\mid\boldsymbol{x}_1),\,g\circ(\mathrm{id}\times\phi^{-1}),\,h\bigr)$
belongs to $\mathbb{E}_{\mathcal{F}_{\mathrm{all}}}(\mathcal{M})$ and induces the same joint law of $(\boldsymbol{z},\boldsymbol{y})$. The latent environment is identified at most up to a smooth invertible reparameterization.
\item For every invertible linear map $\boldsymbol{A}$ on the representation space, the model with representation $\widetilde{\boldsymbol{z}}=\boldsymbol{A}\boldsymbol{z}$ and decoder $h\circ\boldsymbol{A}^{-1}$ belongs to $\mathbb{E}_{\mathcal{F}_{\mathrm{all}}}(\mathcal{M})$. If $\boldsymbol{A}$ couples coordinates of $\boldsymbol{z}_2$ into the first $\dim(\boldsymbol{z}_1)$ coordinates and the coupled part is not almost surely constant in $\boldsymbol{\varepsilon}$, the stable component of $\widetilde{\boldsymbol{z}}$ depends on $\boldsymbol{\varepsilon}$.
\end{enumerate}
Consequently, membership in $\mathbb{E}_{\mathcal{F}_{\mathrm{all}}}(\mathcal{M})$ determines neither the latent environment nor which coordinates of the representation are invariant to it.
\end{proposition}

Appendix~\ref{app:proof_entangle} gives the proof. The equivalence is exact and persists at every sample size, so no amount of data removes it. We call the resulting ambiguity \emph{preference entanglement}. Under it, inductive bias decides the recovered stable--sensitive split. Identifiability requires a model class that rules out the two transformations above and still lets preference depend on the environment.

\begin{definition}[Conditionally identifiable risk-aware recommendation]
\label{def:cirr}
Let $\approx$ denote equality up to permutation and component-wise invertible transformation of the environment-sensitive coordinates. Given a model class $\mathcal{F}$, the CI-RR problem is
\begin{equation}
\label{eq:cirr}
\begin{aligned}
\min_{\mathcal{M}\in\mathcal{F}}\;\; &\mathcal{R}_{\mathrm{te}}(f_{\mathcal{M}})\\
\text{subject to}\;\;
&\mathbb{E}_{\mathcal{F}}(\mathcal{M})
\subseteq\{\widetilde{\mathcal{M}}\in\mathcal{F}:\widetilde{\mathcal{M}}\approx\mathcal{M}\}.
\end{aligned}
\end{equation}
\end{definition}

CI-RR measures expected deployment risk under Assumption~\ref{as:shift}. Worst-case distributionally robust optimization lies outside its scope. The constraint asks that every model observationally equivalent to $\mathcal{M}$ agree with it up to $\approx$, which collapses each equivalence class to a single stable--sensitive decomposition. The observed distribution and the model class then fix the recovered environment-sensitive representation. The constraint involves unknown population distributions and cannot be checked directly, so we construct restrictions on $\mathcal{F}$ that imply it. Section~\ref{sec:method} defines the restricted class, and Section~\ref{sec:theory} proves identifiability under sufficient feature variation.

\begin{remark}[Interpretation of the restriction]
\label{rem:restriction}
Definition~\ref{def:cirr} excludes transformations that change the stable--sensitive decomposition while preserving the interaction distribution. We state the requirement as an inclusion. The reverse inclusion would only say that $\mathcal{F}$ is closed under $\approx$, and risk transfer does not need it. Section~\ref{sec:theory} establishes identifiability under explicit conditions. Section~\ref{sec:controlled_claims} measures the effects of misspecifying the conditional environment and preference mechanism.
\end{remark}

\section{CILER}
\label{sec:method}

Proposition~\ref{prop:entangle} leaves two objects free, namely the conditional environment law in part (i) and the preference mechanism in part (ii). CILER restricts both. A conditional exponential family ties the natural parameters of the environment to the observed user context, which removes the reparameterization freedom of part (i). A polynomial structural model ties feature-indexed coefficients to directed dependencies among environment-sensitive preferences, which removes the mixing freedom of part (ii). The two restrictions define the class
\begin{equation}
\label{eq:fCILER}
\mathcal{F}_{\mathrm{CILER}}
=\left\{
\mathcal{M}:
\begin{aligned}
&q_{\varphi_1}(\boldsymbol{\varepsilon}\mid\boldsymbol{x}_1)\in\mathcal{Q}_{\mathrm{EF}},\\
&\boldsymbol{z}'_{2,i}=g_i\!\left(\operatorname{pa}(\boldsymbol{z}_{2,i});
\boldsymbol{\Lambda}_i(\boldsymbol{x}_1)\right)+\boldsymbol{\varepsilon}_i
\end{aligned}
\right\},
\end{equation}
where $\mathcal{Q}_{\mathrm{EF}}$ is an exponential family with feature-dependent natural parameters and $g_i$ is a polynomial with feature-dependent coefficients.

Fig.~\ref{framework} organizes CILER into three modules. \emph{Conditional environment inference} estimates the environment distribution from observed user features. \emph{Preference mechanism estimation} separates stable and environment-sensitive preferences and structures how the environment enters the latter. \emph{Environment-marginalized prediction} integrates item probabilities over the inferred distribution.

\begin{figure*}[!tb]
\centering
\includegraphics[width=1.0\textwidth]{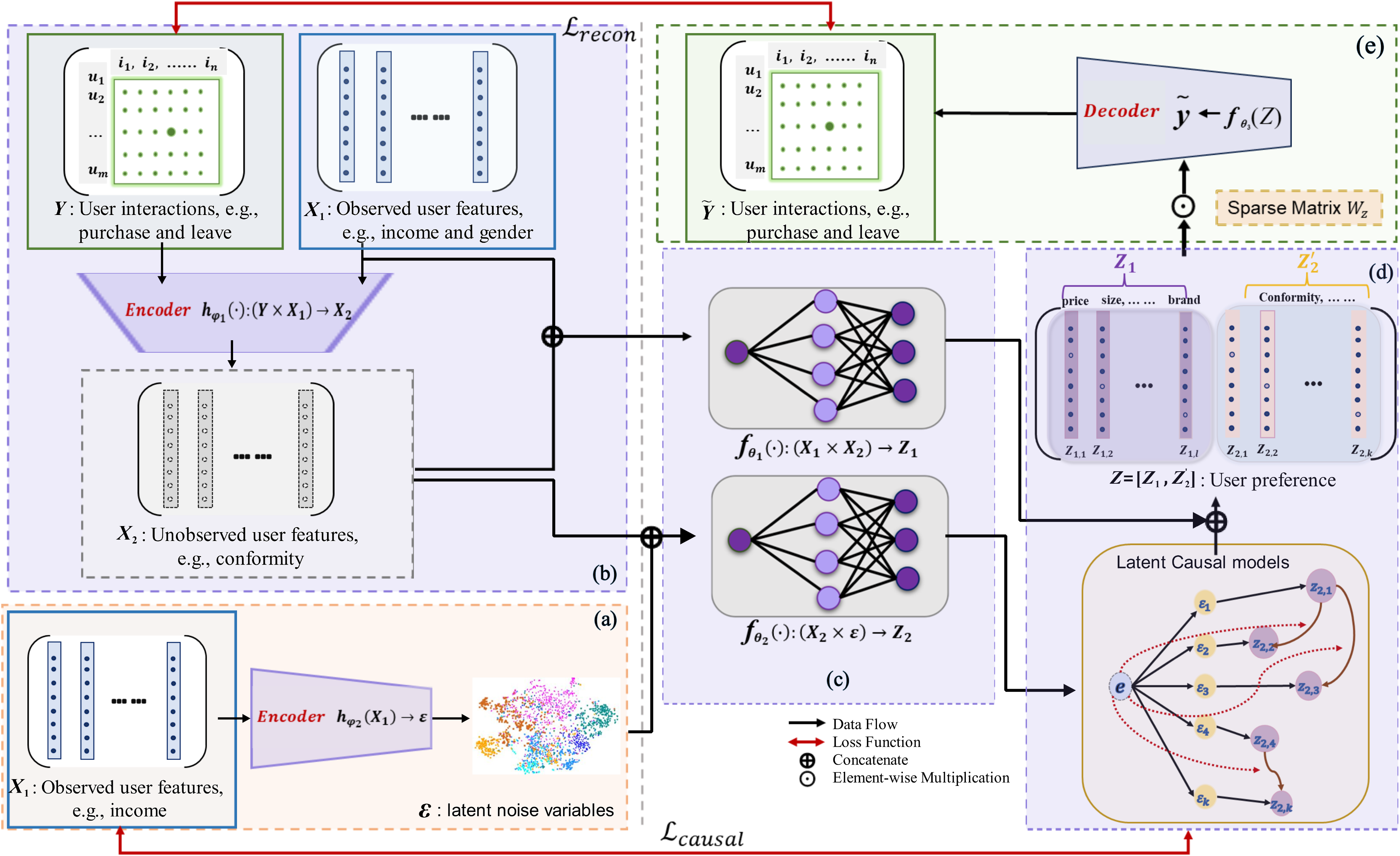}
\caption{CILER framework. Observed user features parameterize the conditional environment. The polynomial mechanism separates stable and environment-sensitive preferences. Inference marginalizes predictions over the learned conditional environment.}
\label{framework}
\end{figure*}

\subsection{Latent generative model}

CILER uses the generative model of Fig.~\ref{causalgraph} to separate user-specific from environment-sensitive sources of variation. Writing $d_2$ for the dimension of the unobserved user features and $n_u$ for the number of positive interactions of user $u$, the sampling process is
\begin{equation}
\left\{
\begin{aligned}
&\boldsymbol{\varepsilon} \sim \operatorname{EF}\!\left(\boldsymbol{\eta}_{\theta_{\varepsilon}}(\boldsymbol{x}_1)\right), \\
&\boldsymbol{x}_2 \sim \mathcal{N}\!\left(\mathbf{0}, \mathbf{I}_{d_2}\right), \\
&\boldsymbol{z}_1 \sim \mathcal{N}\!\left(
\boldsymbol{\mu}_{\theta_1}(\boldsymbol{x}_1,\boldsymbol{x}_2),
\mathrm{diag}\!\left(\boldsymbol{\sigma}^2_{\theta_1}(\boldsymbol{x}_1,\boldsymbol{x}_2)\right)
\right), \\
&\boldsymbol{z}_2 \sim \mathcal{N}\!\left(
\boldsymbol{\mu}_{\theta_2}(\boldsymbol{x}_2,\boldsymbol{\varepsilon}),
\mathrm{diag}\!\left(\boldsymbol{\sigma}^2_{\theta_2}(\boldsymbol{x}_2,\boldsymbol{\varepsilon})\right)
\right), \\
&\boldsymbol{y} \sim \mathrm{Multinomial}\!\left(
n_u;\, \boldsymbol{\pi}\!\left(f_{\theta_3}(\boldsymbol{z}_1,\boldsymbol{z}_2)\right)
\right).
\end{aligned}
\right.
\label{sampling}
\end{equation}
Separating the inputs assigns statistical roles to $\boldsymbol{z}_1$ and $\boldsymbol{z}_2$ without identifying them. Fig.~\ref{causalgraph} summarizes the model. Proposition~\ref{prop:entangle} still permits equivalent reparameterizations of this wiring. The next two subsections add the restrictions that remove them.

\begin{center}
\includegraphics[width=1.0\columnwidth]{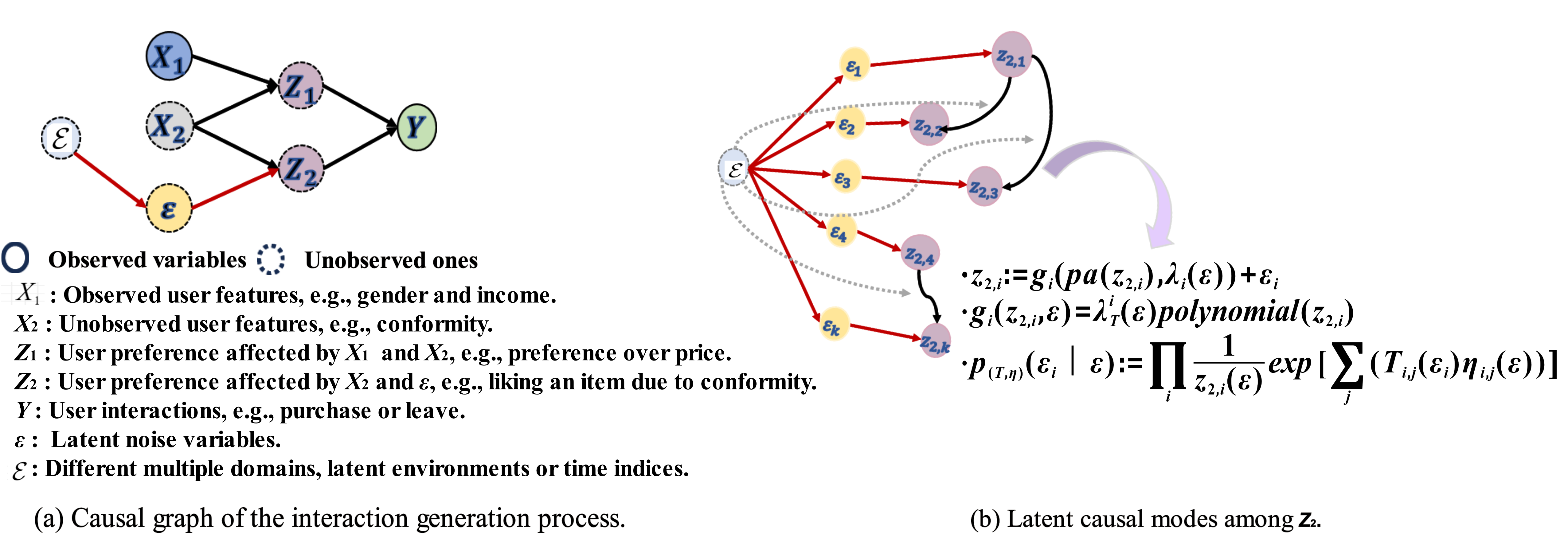}
\captionsetup{hypcap=false}
\captionof{figure}{CILER representation. Observed user context $\boldsymbol{x}_1$ indexes the latent-environment distribution. Unobserved user features $\boldsymbol{x}_2$ capture residual user-specific variation. The stable representation $\boldsymbol{z}_1$ receives no direct input from the latent environment. The environment-sensitive representation $\boldsymbol{z}_2$ follows the constrained preference mechanism. At deployment, CILER infers the latent variables from interaction history and observed user context and marginalizes item probabilities over the inferred latent-environment distribution.}
\label{causalgraph}
\end{center}

\subsection{Conditional environment modeling}\label{sec:env}
The conditional environment has to represent user-dependent variation without any environment label. CILER models it with an exponential family whose natural parameters are functions of the observed user features.
\begin{equation}
q_{\varphi_1}(\boldsymbol{\varepsilon}\mid\boldsymbol{x}_1)
\propto b(\boldsymbol{\varepsilon})
\exp\!\left(
\boldsymbol{\eta}_{\varphi_1}(\boldsymbol{x}_1)^\top
\boldsymbol{T}(\boldsymbol{\varepsilon})
\right).
\label{eq:noise_prior}
\end{equation}
Here $b$ is the base measure and $\boldsymbol{T}$ is the sufficient statistic. The map $\boldsymbol{x}_1\mapsto\boldsymbol{\eta}_{\varphi_1}(\boldsymbol{x}_1)$ turns observed context into a finite-dimensional description of the environment. The identifiability analysis of Section~\ref{sec:theory} draws on the variation of this map across users. Equation~\eqref{eq:noise_prior} defines the variational estimator $q_{\varphi_1}(\boldsymbol{\varepsilon}\mid\boldsymbol{x}_1)$ of the generative law $p_{\theta_{\varepsilon}}(\boldsymbol{\varepsilon}\mid\boldsymbol{x}_1)$, and it uses only features available at deployment. The population result assumes $q_{\varphi_1}=p_{\theta_{\varepsilon}}=p^{*}$. We report the Gaussian, Gamma, Beta, Dirichlet, and Exponential families. Theorem~\ref{thm:ident} covers the families that factorize across the coordinates of $\boldsymbol{\varepsilon}$, and Section~\ref{sec:sensitivity} reports the Dirichlet and discrete cases as diagnostics.

\subsection{Preference mechanism estimation}\label{sec:mech}
CILER first infers the unobserved user features from interaction history.
\begin{equation}
q_{\varphi_2}(\boldsymbol{x}_2\mid\boldsymbol{y},\boldsymbol{x}_1)
=\mathcal{N}\!\left(
\boldsymbol{\mu}_{\varphi_2}(\boldsymbol{y},\boldsymbol{x}_1),
\operatorname{diag}\!\left(\boldsymbol{\sigma}_{\varphi_2}^2
(\boldsymbol{y},\boldsymbol{x}_1)\right)
\right).
\label{eq:latent_user_posterior}
\end{equation}
Separating $\boldsymbol{x}_2$ from $\boldsymbol{\varepsilon}$ keeps individual taste apart from environmental variation. The two preference components receive disjoint inputs. The stable component $\boldsymbol{z}_1$ depends on $(\boldsymbol{x}_1,\boldsymbol{x}_2)$, and the environment-sensitive component $\boldsymbol{z}_2$ depends on $(\boldsymbol{x}_2,\boldsymbol{\varepsilon})$.
\begin{equation}
\begin{aligned}
p(\boldsymbol{z}_1\mid\boldsymbol{x}_1,\boldsymbol{x}_2)
&=\mathcal{N}\!\left(\boldsymbol{\mu}_{\theta_1},
\operatorname{diag}(\boldsymbol{\sigma}_{\theta_1}^2)\right),\\
p(\boldsymbol{z}_2\mid\boldsymbol{x}_2,\boldsymbol{\varepsilon})
&=\mathcal{N}\!\left(\boldsymbol{\mu}_{\theta_2},
\operatorname{diag}(\boldsymbol{\sigma}_{\theta_2}^2)\right),
\end{aligned}
\label{user-pref}
\end{equation}
where the moments are produced by neural encoders with the corresponding inputs.

\noindent\textbf{Polynomial preference mechanism.}
The amortized posterior in Eq.~\eqref{user-pref} still admits the transformations of Proposition~\ref{prop:entangle}. CILER restricts how the environment enters $\boldsymbol{z}_2$ through a structural model with feature-modulated coefficients~\citep{liu2024identifiable,locatello2019challenging}.
\begin{equation}
\boldsymbol{z}_{2,i}'
= g_i\!\left(\operatorname{pa}(\boldsymbol{z}_{2,i});
\boldsymbol{\Lambda}_i(\boldsymbol{x}_1)\right)
+\boldsymbol{\varepsilon}_i,
\quad i=1,\ldots,k.
\label{causalZ2}
\end{equation}
The function $g_i$ is a polynomial in the parents of $\boldsymbol{z}_{2,i}$, and $\boldsymbol{\Lambda}_i(\boldsymbol{x}_1)$ contains its coefficients. The amortized sample $\boldsymbol{z}_{2,i}$ comes from Eq.~\eqref{user-pref}. Eq.~\eqref{causalZ2} defines its structural counterpart $\boldsymbol{z}_{2,i}'$. The fit term aligns these representations during training. A lower-triangular parameterization of $\boldsymbol{\Lambda}(\boldsymbol{x}_1)$ enforces acyclicity under a reference ordering of the latent coordinates~\citep{zheng2018dags}. The ordering fixes only the direction of admissible edges, and which latent factor occupies each coordinate is recovered by learning. The structural equation model (SEM) loss combines mechanism fit, latent-user regularization, conditional-environment alignment, and coefficient sparsity.
\begin{equation}
\begin{aligned}
\mathcal{L}_{\mathrm{causal}}
={}&\mathcal{L}_{\mathrm{SEM}}\\&+\mathrm{KL}\!\left(q_{\varphi_2}(\boldsymbol{x}_2\mid\boldsymbol{y},\boldsymbol{x}_1)\,\|\,p(\boldsymbol{x}_2)\right)
+\sum_{i=1}^{k}\mathrm{KL}\!\Bigl(
q_{\varphi_1}(\boldsymbol{\varepsilon}_i\mid\boldsymbol{x}_1)\\
&\quad\|\,
p_{\theta_{\varepsilon}}(\boldsymbol{\varepsilon}_i\mid\boldsymbol{x}_1)
\Bigr)
+\sum_{i>j}\|\Lambda_{ij}(\boldsymbol{x}_1)\|_1.
\end{aligned}
\label{causalloss}
\end{equation}
Appendix~\ref{app:structural_details} gives the polynomial basis, the reference-order parameterization, and the form of $\mathcal{L}_{\mathrm{SEM}}$.

\noindent\textbf{Why the two restrictions are coordinated.}
The two restrictions are indexed by the same variable, namely the observed features $\boldsymbol{x}_1$. The exponential family turns a change in $\boldsymbol{x}_1$ into a change of natural parameters. The polynomial mechanism turns the same change into a change of structural coefficients. Neither restriction suffices on its own. Restricting only the environment law still leaves the linear mixing of Proposition~\ref{prop:entangle}(ii), and restricting only the mechanism still leaves the reparameterization of Proposition~\ref{prop:entangle}(i). Together they force the two changes to move consistently, and Assumption~\ref{as:variation} converts that consistency into identifiability. Section~\ref{sec:ablation} removes each restriction in turn.

\begin{algorithm}[t]
\caption{CILER}
\label{alg:training}
\begin{algorithmic}[1]
\Require Training interactions $\mathcal{D}=\{(\boldsymbol{y}_i,\boldsymbol{x}_{1,i})\}_{i=1}^{N}$, test user $(\boldsymbol{y}_*,\boldsymbol{x}_{1,*})$, loss weights $\lambda,\alpha,\beta$, environment sample size $S$
\Ensure Model parameters $\Theta$, item distribution $\widehat{\boldsymbol{\pi}}_*$
\State Initialize $\Theta\gets\{\varphi_1,\varphi_2,\theta_1,\theta_2,\theta_3,\theta_{\varepsilon}\}$ and decoder gate $\boldsymbol{W}_z$.
\While{$\Theta$ has not converged}
\State Sample a mini-batch $\mathcal{B}\subset\mathcal{D}$.
\State $\boldsymbol{\varepsilon}\gets q_{\varphi_1}(\cdot\mid\boldsymbol{x}_1)$ by Eq.~\eqref{eq:noise_prior}.
\State $\boldsymbol{x}_2\gets q_{\varphi_2}(\cdot\mid\boldsymbol{y},\boldsymbol{x}_1)$ by Eq.~\eqref{eq:latent_user_posterior}.
\State $\boldsymbol{z}_1\gets p(\cdot\mid\boldsymbol{x}_1,\boldsymbol{x}_2)$ and $\boldsymbol{z}_2\gets p(\cdot\mid\boldsymbol{x}_2,\boldsymbol{\varepsilon})$ by Eq.~\eqref{user-pref}.
\State $\boldsymbol{z}_2^{\prime}\gets g(\operatorname{pa}(\boldsymbol{z}_2);\boldsymbol{\Lambda}(\boldsymbol{x}_1))+\boldsymbol{\varepsilon}$ by Eq.~\eqref{causalZ2}.
\State $\boldsymbol{z}_{\mathrm{dec}}\gets\boldsymbol{W}_z\odot[\boldsymbol{z}_1,\boldsymbol{z}_2^{\prime}]$ and $\boldsymbol{\pi}\gets\operatorname{softmax}(f_{\theta_3}(\boldsymbol{z}_{\mathrm{dec}}))$ by Eq.~\eqref{eq:decoder_gate}.
\State Compute $\mathcal{L}_{\mathrm{total}}$ on $\mathcal{B}$ by Eq.~\eqref{loss}.
\State Update $(\Theta,\boldsymbol{W}_z)$ by descending $\nabla\mathcal{L}_{\mathrm{total}}$.
\EndWhile
\State $\boldsymbol{x}_{2,*}\gets q_{\varphi_2}(\cdot\mid\boldsymbol{y}_*,\boldsymbol{x}_{1,*})$ by Eq.~\eqref{eq:latent_user_posterior}.
\State $\boldsymbol{z}_{1,*}\gets p(\cdot\mid\boldsymbol{x}_{1,*},\boldsymbol{x}_{2,*})$ by Eq.~\eqref{user-pref}.
\For{$s=1,\ldots,S$}
\State $\boldsymbol{\varepsilon}^{(s)}_*\gets q_{\varphi_1}(\cdot\mid\boldsymbol{x}_{1,*})$ by Eq.~\eqref{eq:noise_prior}.
\State $\boldsymbol{z}_{2,*}^{\prime(s)}\gets$ Eqs.~\eqref{user-pref} and~\eqref{causalZ2} given $\boldsymbol{\varepsilon}^{(s)}_*$.
\State Compute $\boldsymbol{\pi}^{(s)}_*$ by Eq.~\eqref{eq:decoder_gate}.
\EndFor
\State $\widehat{\boldsymbol{\pi}}_*\gets S^{-1}\sum_{s=1}^{S}\boldsymbol{\pi}^{(s)}_*$ by Eq.~\eqref{eq:ood_decomposition}.
\State \Return $\Theta,\widehat{\boldsymbol{\pi}}_*$
\end{algorithmic}
\end{algorithm}

\noindent\textbf{Decoding.}
A shared element-wise gate $\boldsymbol{z}_{\mathrm{dec}}=\boldsymbol{W}_z\odot[\boldsymbol{z}_1,\boldsymbol{z}_2']$ regularizes the decoder input with penalty $\alpha\|\boldsymbol{W}_z\|_1$. The gate is an auxiliary component and plays no role in Theorem~\ref{thm:ident}. Appendix~\ref{app:gate_analysis} analyzes it separately.
\begin{equation}
\label{eq:decoder_gate}
\boldsymbol{\pi}=\operatorname{softmax}\!\left(
f_{\theta_3}(\boldsymbol{z}_{\mathrm{dec}})\right),
\quad
\mathcal{L}_{\mathrm{recon}}
=-\mathbb{E}\!\left[\sum_{i\in\mathcal{I}^{+}}\log \pi_i\right],
\end{equation}
where $\mathcal{I}^{+}$ denotes the positive items of the user.

\subsection{Training objective and environment-marginalized prediction}\label{sec:predict}
CILER learns the conditional environment and the preference mechanism jointly by minimizing
\begin{equation}
\mathcal{L}_{\mathrm{total}}
=\mathcal{L}_{\mathrm{recon}}
+\lambda\mathcal{L}_{\mathrm{causal}}
+\alpha\|\boldsymbol{W}_z\|_1
+\beta\|\Theta\|_2^2.
\label{loss}
\end{equation}
Here $\lambda$, $\alpha$, and $\beta$ weight the structural loss, decoder gate, and weight decay. Appendix~\ref{app:variational_objective} derives the reconstruction and KL terms from the evidence lower bound. Eq.~\eqref{loss} is a trainable surrogate for the objective of the CI-RR problem in Eq.~\eqref{eq:cirr}. The choice of the class $\mathcal{F}_{\mathrm{CILER}}$ meets the identifiability constraint (Theorem~\ref{thm:ident}), and the loss itself does not enforce it. Training uses no samples from $P_{\mathrm{te}}$.

At deployment the environment is still a conditional distribution given the observed user features. In the model of Eq.~\eqref{sampling}, $\boldsymbol{\varepsilon}$ is drawn from $\boldsymbol{x}_1$ alone, so it is independent of $\boldsymbol{x}_2$ given $\boldsymbol{x}_1$. The predictive representation becomes
\begin{equation}
\label{eq:ood_decomposition}
P(\boldsymbol{z}\mid\boldsymbol{x}_1,\boldsymbol{x}_2)
=\int_{\mathcal{E}}
P(\boldsymbol{z}\mid\boldsymbol{x}_1,\boldsymbol{x}_2,\boldsymbol{\varepsilon})
P(\boldsymbol{\varepsilon}\mid\boldsymbol{x}_1)
\,d\boldsymbol{\varepsilon}.
\end{equation}
CILER predicts with the corresponding Monte Carlo average. It first infers $\boldsymbol{x}_2$ and $\boldsymbol{z}_1$ for the observed user. It then draws $S$ environments, propagates each through the preference mechanism, and averages the item probabilities.
\begin{equation}
\widehat{\boldsymbol{\pi}}(\boldsymbol{y},\boldsymbol{x}_1)
=\frac{1}{S}\sum_{s=1}^{S}
\operatorname{softmax}\!\left(
f_{\theta_3}\!\left(
\boldsymbol{W}_z\odot[\boldsymbol{z}_1,\boldsymbol{z}_2'^{(s)}]
\right)\right).
\label{eq:environment_marginal_inference}
\end{equation}
CILER handles the two sources of uncertainty separately. The interaction history resolves the user-specific part through Eq.~\eqref{eq:latent_user_posterior}. Equation~\eqref{eq:environment_marginal_inference} integrates out the environmental part. Proposition~\ref{prop:risk} bounds the logarithmic risk of the resulting predictive distribution.

This construction gives CILER three operational properties.
\begin{itemize}
\item \textbf{Label-free inference.} CILER infers environments from observed user features without partition or segment annotations.
\item \textbf{Identifiable representation.} Under the conditions of Theorem~\ref{thm:ident}, the model class determines the environment-sensitive representation up to permutation and component-wise invertible transformation.
\item \textbf{Update-free prediction.} Marginalization adds $S$ forward passes through the sensitive branch and uses no test-time gradient updates.
\end{itemize}

Algorithm~\ref{alg:training} summarizes training and inference. The analysis below establishes identifiability within $\mathcal{F}_{\mathrm{CILER}}$ and bounds the excess deployment risk of environment marginalization.

\section{Theoretical analysis}
\label{sec:theory}
The analysis follows the order of the problem. Proposition~\ref{prop:entangle} has already shown why restrictions are needed. Theorem~\ref{thm:ident} shows that the restricted class recovers the environment-sensitive representation and its directed relations. Proposition~\ref{prop:risk} then bounds the excess deployment risk of environment marginalization.

\subsection{Conditions}
Proposition~\ref{prop:entangle} establishes non-identifiability in the unrestricted class. The following assumptions define when the coordinated restrictions recover the environment-sensitive representation. Assumption~\ref{as:variation} concerns observed feature variation. Assumptions~\ref{as:spec} and~\ref{as:decoder} restrict the model class. Assumption~\ref{as:shift} enters only when the recovered model is used at deployment. Recent identifiability results likewise use variation across environments or observed conditions to restrict the latent model~\citep{ng2025general,liu2026weightvariant}.

The first assumption asks the observed context to induce enough independent parameter changes to separate the latent coordinates. Let $m$ be the dimension of the per-coordinate sufficient statistic, so that $\boldsymbol{\eta}_{\varphi_1}$ has $mk$ entries.

\begin{assumption}[Sufficient variation]
\label{as:variation}
The conditional environment factorizes across the coordinates of $\boldsymbol{\varepsilon}$. There exist $mk+1$ feature values $\boldsymbol{x}_1^{(0)},\ldots,\boldsymbol{x}_1^{(mk)}$ such that the differences $\boldsymbol{\eta}_{\varphi_1}(\boldsymbol{x}_1^{(j)})-\boldsymbol{\eta}_{\varphi_1}(\boldsymbol{x}_1^{(0)})$ for $j=1,\ldots,mk$ are linearly independent, and the structural coefficients $\boldsymbol{\Lambda}(\boldsymbol{x}_1)$ also vary over these values. This is the sufficient-change condition of~\citet{liu2024identifiable}.
\end{assumption}

The second assumption asks the selected environment family, the polynomial degree, and the structural ordering to be correctly specified.

\begin{assumption}[Specification]
\label{as:spec}
The conditional environment belongs to the regular continuous exponential family in Eq.~\eqref{eq:noise_prior} and uses the sufficient statistic specified in estimation. The true preference mechanism is a polynomial of degree at most $p$. Its causal ordering agrees with the reference ordering of $\boldsymbol{\Lambda}(\boldsymbol{x}_1)$.
\end{assumption}

The third assumption links the observed interactions to the latent representation.

\begin{assumption}[Decoder regularity]
\label{as:decoder}
The map from the representation $\boldsymbol{z}$ to the item logits $f_{\theta_3}(\boldsymbol{z})$ is smooth and injective on the support of $\boldsymbol{z}$, and satisfies the smoothness conditions of the identifiability theorem of~\citet{liu2024identifiable}. In addition, the law of the item logits is determined by the multinomial interaction distribution $p(\boldsymbol{y}\mid\boldsymbol{x}_1)$ up to an additive constant per user.
\end{assumption}

\subsection{Identifiability within the restricted class}

\begin{theorem}[Conditional identifiability of CILER]
\label{thm:ident}
Let Assumptions~\ref{as:variation}--\ref{as:decoder} hold and let $\mathcal{M},\widetilde{\mathcal{M}}\in\mathcal{F}_{\mathrm{CILER}}$ be observationally equivalent in the sense of Definition~\ref{def:oec}. Then their environment-sensitive representations coincide up to permutation and component-wise invertible transformation, and the directed relations among these variables coincide up to the same equivalence. Consequently
$\mathbb{E}_{\mathcal{F}_{\mathrm{CILER}}}(\mathcal{M})
\subseteq\{\widetilde{\mathcal{M}}\in\mathcal{F}_{\mathrm{CILER}}:\widetilde{\mathcal{M}}\approx\mathcal{M}\}$,
so $\mathcal{F}_{\mathrm{CILER}}$ satisfies the constraint of the CI-RR problem in Eq.~\eqref{eq:cirr}. If Assumption~\ref{as:shift} holds in addition, the recovered model transfers to the deployment distribution.
\end{theorem}

Appendix~\ref{app:proof_ident} provides the proof by mapping CILER to the latent polynomial structural model of~\citet{liu2024identifiable}. Observed user features serve as the auxiliary variable. Their variation excludes smooth environment reparameterizations that conflict with the feature-indexed polynomial mechanism. The reference structural ordering excludes transformations that mix environment-sensitive coordinates into the stable representation. Permutation and component-wise invertible transformations remain. These transformations preserve the marginalized item probabilities in Eq.~\eqref{eq:environment_marginal_inference}.

\begin{remark}[Scope of the guarantee]
\label{rem:scope}
Three conditions produce the identifiability claim, and a fourth carries it to deployment. Assumption~\ref{as:variation} requires $mk$ effective directions of feature-induced parameter change and a conditional environment that factorizes across coordinates. Assumption~\ref{as:spec} restricts the environment to a regular continuous family and the mechanism to degree $p$, so the discrete-family diagnostics in Section~\ref{sec:sensitivity} fall outside the guarantee. Assumption~\ref{as:decoder} requires a smooth injective map from representations to item logits and, separately, recoverability of those logits from the multinomial likelihood. The second requirement is stronger than the likelihood alone~\citep{hyvarinen2019nonlinear,khemakhem2020variational}. Assumption~\ref{as:shift} enters only the transfer step, where it limits deployment features to the training support under a bounded density ratio. Section~\ref{sec:controlled_claims} tests the observable consequences of Assumptions~\ref{as:variation} and~\ref{as:spec} on controlled data. Real-data experiments evaluate downstream prediction within the stated support.
\end{remark}

\subsection{Risk control by environment marginalization}
Identifiability fixes what the model recovers. The next result fixes how the recovered conditional environment enters prediction. Environmental uncertainty remains after conditioning on the observed features, and it contributes to deployment risk.

\begin{proposition}[Risk transfer under environment marginalization]
\label{prop:risk}
Let Assumption~\ref{as:shift} hold, let $p^{*}(\boldsymbol{\varepsilon}\mid\boldsymbol{x}_1)$ be the true conditional environment, and let
$p_{q}(\boldsymbol{y}\mid\boldsymbol{x}_1)
=\int p(\boldsymbol{y}\mid\boldsymbol{x}_1,\boldsymbol{\varepsilon})\,
q(\boldsymbol{\varepsilon}\mid\boldsymbol{x}_1)\,d\boldsymbol{\varepsilon}$
denote the environment-marginalized predictive law obtained from an inferred conditional environment $q$. Under the logarithmic loss, the excess deployment risk of $p_q$ relative to the oracle predictive law $p^{*}$ satisfies
\begin{align}
\label{eq:risk_transfer}
\mathcal{R}_{\mathrm{te}}(p_{q})-\mathcal{R}_{\mathrm{te}}(p^{*})
&=\mathbb{E}_{\boldsymbol{x}_1\sim P_{\mathrm{te}}}
\mathrm{KL}\!\left(p^{*}(\boldsymbol{y}\mid\boldsymbol{x}_1)\,\|\,p_{q}(\boldsymbol{y}\mid\boldsymbol{x}_1)\right) \nonumber\\
&\leq
\mathbb{E}_{\boldsymbol{x}_1\sim P_{\mathrm{te}}}
\mathrm{KL}\!\left(p^{*}(\boldsymbol{\varepsilon}\mid\boldsymbol{x}_1)\,\|\,q(\boldsymbol{\varepsilon}\mid\boldsymbol{x}_1)\right).
\end{align}
Moreover, if $q(\cdot\mid\boldsymbol{x}_1)=\delta_{\widehat{\boldsymbol{\varepsilon}}(\boldsymbol{x}_1)}$ is a single environment state and $p^{*}(\cdot\mid\boldsymbol{x}_1)$ is non-degenerate on a set of positive probability, the right-hand side of Eq.~\eqref{eq:risk_transfer} is infinite.
\end{proposition}

Equation~\eqref{eq:risk_transfer} transfers conditional environment-estimation error to logarithmic prediction risk on the deployment support. The bound is stated under $P_{\mathrm{te}}$, which is not observed during training. Assumption~\ref{as:shift} turns it into a training-side quantity through the bounded density ratio, which gives $\mathbb{E}_{P_{\mathrm{te}}}[\mathrm{KL}]\le B\,\mathbb{E}_{P_{\mathrm{tr}}}[\mathrm{KL}]$. Lower environment-inference error on the training features tightens the deployment bound. A deterministic $q$ against a non-degenerate $p^{*}(\boldsymbol{\varepsilon}\mid\boldsymbol{x}_1)$ makes the KL term infinite and the bound uninformative.

The three results raise separate empirical questions. Proposition~\ref{prop:entangle} motivates the restrictions, Theorem~\ref{thm:ident} states the conditions for recovery, and Proposition~\ref{prop:risk} governs environment-marginalized prediction. Its bound concerns logarithmic loss, so Section~\ref{sec:marginalization} reports negative log-likelihood and Brier score alongside Recall and NDCG.

\section{Experiments}
\label{sec:exp}
In this section, we present the experimental results to validate two questions: (i) Does CILER improve OOD ranking under the three shared-support shift protocols while retaining competitive IID accuracy? (ii) Do the two restrictions and environment marginalization account for the gains? Main comparisons address the first question. Capacity-matched ablations and paired predictive diagnostics address the second. Controlled studies test recovery when the observable variation and specification conditions hold. Sensitivity and complexity analyses define the operating range and computational cost.

\subsection{Setup}
\label{sec:setup}

\noindent\textbf{Datasets and shift protocols.}
We use the Synthetic, Meituan, and Yelp protocols released with COR~\citep{wang2022causal} and adopted by DT3OR~\citep{yang2025dual}. Synthetic shifts observed user covariates. Meituan separates weekdays from weekends, and Yelp separates geographical regions. In each protocol, $\boldsymbol{x}_1$ contains covariates available during training and deployment. The split variable constructs the evaluation sets and is never supplied to CILER as an environment label. We treat each protocol as a marginal shift in observed context within shared support. Real-data experiments evaluate prediction under these protocols, and controlled experiments assess the recovery assumptions. Ratings of four or higher define positive interactions. We split each dataset into training, validation, and test sets using proportions of $80\%$, $10\%$, and $10\%$. Table~\ref{tab:Statistics} reports the dataset statistics.

\begin{table}
\centering
\caption{Statistics of the three datasets. IID and OOD denote in-distribution and out-of-distribution interactions.}
\label{tab:Statistics}
\resizebox{\linewidth}{!}{%
\begin{tabular}{lccccc}
\toprule \textbf{Dataset} & \textbf{\#User} & \textbf{\#Item} & \textbf{\#IID interactions} & \textbf{\#OOD interactions} & \textbf{Density} \\
\midrule  \textbf{Synthetic} & 1,000 & 1,000 & 145,270 & 112,371 & 0.257641 \\
 \textbf{Meituan} & 2,145 & 7,189 & 11,400 & 6,944 & 0.001189 \\
 \textbf{Yelp} & 7,975 & 74,722 & 305,128 & 99,525 & 0.000679 \\
\bottomrule
\end{tabular}%
}
\end{table}

\noindent\textbf{Baselines.}
We compare CILER with thirteen baselines spanning feature interaction, variational recommendation, causal representation learning, invariant learning, debiasing, distributionally robust learning, diffusion recommendation, and test-time adaptation. The baselines are factorization machines (FM)~\citep{rendle2010factorization}, neural factorization machines (NFM)~\citep{he2017neural}, MultiVAE~\citep{liang2018variational}, MacridVAE and MacridVAE+FM~\citep{ma2019learning}, CausPref~\citep{he2022causpref}, COR~\citep{wang2022causal}, InvCF~\citep{zhang2023invariant}, CDR~\citep{wang2023causal}, PopGo~\citep{zhang2024popshift}, DR-GNN~\citep{wang2024distributionally}, CausalDiffRec~\citep{zhao2025graph}, and DT3OR~\citep{yang2025dual}. Together they cover external environment partitions, designated shift variables, worst-case optimization, and deployment-time adaptation. Methods use their published input regimes. CDR receives temporal segment labels. DT3OR uses deployment interactions and test-time gradient updates. CILER uses observed user context without environment labels or test-time updates. Appendix~\ref{app:experimental_details} describes each method and its search space.

\noindent\textbf{Metrics.}
Recall@$K$ measures retrieval coverage. Normalized discounted cumulative gain at $K$ (NDCG@$K$) measures ranking quality~\citep{gao2023cirs}. Mean reciprocal rank (MRR) reports the reciprocal rank of the first relevant item. Tables use R@$K$ and N@$K$ as compact forms. We compare methods within each dataset. The marginalization study also reports negative log-likelihood (NLL) and Brier score because Proposition~\ref{prop:risk} bounds logarithmic risk.

\noindent\textbf{Training protocol.}
Hyperparameter search and final evaluation use separate epoch budgets. Synthetic and Meituan use a learning rate of 0.001. Yelp uses 0.00075. The tuning stage runs 300 epochs on Synthetic and Meituan and 30 epochs on Yelp. The final runs use 1,000, 300, and 200 epochs, respectively. The final protocol evaluates ten random seeds and reports all ranking values to four decimal places. Controlled comparisons use ten paired seeds with percentile-bootstrap confidence intervals over paired effects. Implementation details are given in Appendix~\ref{app:experimental_details}. Model selection uses the validation split, and the selected configuration is transferred unchanged to the test split. Baselines are tuned on the same validation split over the search spaces listed in Appendix~\ref{app:experimental_details}.

\subsection{Main Results}
\label{sec:results}
\begin{table*}[htbp]
\centering
\caption{OOD performance comparison on diverse benchmarks. The best result is shown in bold, and the second-best result is underlined. CDR requires discrete temporal segment labels, which are unavailable for Synthetic and Yelp.}
\label{tab:results}
\resizebox{\textwidth}{!}{
\begin{tabular}{cccccccccccccc}
\toprule
\multicolumn{2}{c}{\textbf{Dataset}} & \multicolumn{4}{c}{\textbf{Synthetic Data}} & \multicolumn{4}{c}{\textbf{Meituan}} & \multicolumn{4}{c}{\textbf{Yelp}} \\ 
\cmidrule(r){3-6}\cmidrule(lr){7-10}\cmidrule(r){11-14}
\multicolumn{2}{c}{\textbf{Metric}} & \textbf{R@10} & \textbf{R@20} & \textbf{N@10} & \textbf{N@20} & \textbf{R@50} & \textbf{R@100} & \textbf{N@50} & \textbf{N@100} & \textbf{R@50} & \textbf{R@100}  & \textbf{N@50} & \textbf{N@100}\\ \midrule
\textbf{FM} & ICDM 2010& 0.0572 & 0.1074 & 0.0604 & 0.0792 & 0.0121 & 0.0205 & 0.0043 & 0.0057 & 0.0964 & 0.1389 & 0.0313 & 0.0385\\ 
\textbf{NFM} & SIGIR 2017 & 0.0405 & 0.0761 & 0.0438 & 0.0560 & 0.0233 & 0.0354 & 0.0066 & 0.0085 & 0.0829 & 0.1276  & 0.0241 & 0.0316\\ 
\textbf{MultiVAE} & WWW 2018 & 0.0208 & 0.0408 & 0.0172 & 0.0257 & 0.0238 & 0.0368 & 0.0069 & 0.0091 & 0.0365 & 0.0582 & 0.0118 & 0.0154\\ 
\textbf{MacridVAE} &NeurIPS 2019 & 0.0231 & 0.0392 & 0.0192 & 0.0262 & 0.0219 & 0.0364 & 0.0067 & 0.0090 & 0.0408 & 0.0634 & 0.0135 & 0.0174\\ 
\textbf{MacridVAE+FM} & NeurIPS 2019 & 0.0463 & 0.0836 & 0.0513 & 0.0643 & 0.0233 & 0.0364 & 0.0066 & 0.0087 & 0.0407 & 0.0626 & 0.0140 & 0.0178\\ 
\textbf{CausPref} & WWW 2022  & 0.0477 & 0.0541 & 0.0674 & 0.0941 &  0.0344 & 0.0435 & 0.0098 & 0.0121 & 0.1346 & 0.1834 & 0.0414 & 0.0501\\
\textbf{COR} & WWW 2022 & 0.0767 & 0.1443 & 0.0804 & 0.1056 & 0.0368 & 0.0578 & 0.0101 & 0.0135 & 0.1416 & 0.1986&  0.0500 &  0.0595\\
\textbf{InvCF} & WWW 2023  & 0.0764 & 0.1309 & 0.0742 & 0.0845 & 0.0298 & 0.0405 & 0.0074 & 0.0098 & 0.1045 & 0.1186 & 0.0304 & 0.0401\\ 
\textbf{CDR} & WWW 2023  & - & - & - & - & 0.0410 & 0.0569 & \underline{0.0145} & \underline{0.0170} & - & - & - & -\\ 
\textbf{PopGo} & TOIS 2024  & 0.0708 & 0.1324 & 0.0710 & 0.0952 & 0.0295 & 0.0428 & 0.0085 & 0.0120 & 0.1268 & 0.1462 & 0.0357 & 0.0416\\ 
\textbf{DR-GNN} & WWW 2024  & 0.0795 & 0.1421 & 0.0800 & 0.1034 & 0.0372 & 0.0575 & 0.0104 & 0.0135 & 0.1411 & 0.1976 & 0.0487 & 0.0503\\ 
\textbf{CausalDiffRec} & WWW 2025  & 0.0772 & 0.1400 & 0.0809 & 0.1061 & 0.0314 & 0.0529 & 0.0098 & 0.0116 & 0.1198 & 0.1345 & 0.0376 & 0.0438\\ 
\textbf{DT3OR} & TKDE 2025 & \underline{0.0820} & \underline{0.1484} & \underline{0.0842} & \underline{0.1086} & \underline{0.0415} & \underline{0.0583} & 0.0113 & 0.0140 & \underline{0.1441} & \underline{0.1990} & \underline{0.0510} & \underline{0.0602}\\ 
\midrule
\multicolumn{2}{c}{\textbf{CILER}} & \textbf{0.0882} & \textbf{0.1626} & \textbf{0.0915} & \textbf{0.1174} & \textbf{0.0503} & \textbf{0.0732} & \textbf{0.0156} & \textbf{0.0190} & \textbf{0.1486} & \textbf{0.2048}  & \textbf{0.0524} & \textbf{0.0619}\\ 
\bottomrule
\end{tabular}
}
\end{table*}

\noindent\textbf{CILER ranks first on all twelve OOD ranking metrics.}
Table~\ref{tab:results} reports the OOD comparison, and the percentages below are relative improvements over the strongest baseline on the same metric. On Synthetic the reference is DT3OR and the gains range from $7.6\%$ to $9.6\%$. On Meituan the Recall reference is DT3OR, with gains of $21.2\%$ at cutoff $50$ and $25.6\%$ at cutoff $100$. The NDCG reference is CDR, with gains of $7.6\%$ and $11.8\%$. On Yelp the reference is again DT3OR and the gains range from $2.7\%$ to $3.1\%$. CILER ranks first under each of the three shift protocols.

\noindent\textbf{The two strongest baselines each require an input that CILER does not use.}
DT3OR attains the strongest baseline value on ten of the twelve metrics, and it does so with test-time gradient updates on deployment interactions. CDR attains the strongest baseline NDCG on the two Meituan settings, and it requires discrete temporal segment labels that the Synthetic and Yelp protocols do not provide. CILER uses observed user context at inference, performs no parameter update after training, and exceeds both methods wherever they are evaluated. The remaining analyses isolate the two restrictions and environment marginalization.

\begin{table}
\centering\caption{IID performance under the COR evaluation protocol. Baseline values are reproduced from COR~\citep{wang2022causal}. Synthetic reports Recall@20. Meituan and Yelp report Recall@50. The final row gives CILER's absolute difference from the strongest published result on each dataset. Bold denotes the best result overall, and underline denotes the second-best result overall.}
\label{tab:iid_reference}
\setlength{\tabcolsep}{9.5pt}
\begin{tabular}{lccc}\toprule \textbf{Method} & \textbf{Synthetic} & \textbf{Meituan} & \textbf{Yelp} \\ 
\midrule 
FM & 0.3666 & 0.0846 & 0.1228 \\ NFM & 0.3629 & 0.0825 & 0.1222 \\ 
MultiVAE & \underline{0.3693} & 0.1054 & 0.1399 \\ MacridVAE & 0.3573 & \underline{0.1163} & 0.1526 \\ 
MacridVAE+FM & 0.3648 & \textbf{0.1219} & \underline{0.1536} \\ COR & 0.3628 & 0.1159 & \textbf{0.1539} \\ 
\midrule 
\textbf{CILER} & \textbf{0.3709} & 0.1147 & 0.1515 \\
$\Delta$  & +0.0016 & -0.0072 & -0.0024  \\
\bottomrule
\end{tabular}
\end{table}

\noindent\textbf{CILER improves OOD ranking while retaining competitive IID accuracy.}
Table~\ref{tab:iid_reference} evaluates CILER under the in-distribution protocol of COR, against which the published baseline values are reported. CILER ranks first on Synthetic. Its gap from the strongest published result is $0.0072$ on Meituan and $0.0024$ on Yelp. These results show that the OOD gains do not entail a material loss of IID ranking accuracy under the same protocol.

\subsection{Contribution of the two restrictions}
\label{sec:ablation}

\begin{table*}[t]
\centering
\caption{Ablation of the two coordinated restrictions and the auxiliary gate. \textit{w/o LE} removes the user-conditioned latent-environment restriction, \textit{w/o SC} removes the feature-indexed polynomial structural mechanism, and \textit{w/o Gate} removes the auxiliary decoder gate. Every variant keeps the encoder, decoder, latent capacity, and input separation of the full model.}
\label{tab:multi_dataset_comparison}
\resizebox{\textwidth}{!}{%
\begin{tabular}{lcccccccccccc}
\toprule
\multirow{2}{*}{\textbf{Variants}}  &
\multicolumn{4}{c}{\textbf{Synthetic Data}} &
\multicolumn{4}{c}{\textbf{Meituan}} &
\multicolumn{4}{c}{\textbf{Yelp}} \\
\cmidrule(lr){2-5} \cmidrule(lr){6-9} \cmidrule(lr){10-13}
~ & \textbf{R@10} & \textbf{R@20} & \textbf{N@10} & \textbf{N@20} & \textbf{R@50} & \textbf{R@100} & \textbf{N@50} & \textbf{N@100} & \textbf{R@50} & \textbf{R@100} & \textbf{N@50} & \textbf{N@100} \\
\midrule
\textbf{CILER (Full)} & \textbf{0.0882} & \textbf{0.1626} & \textbf{0.0915} & \textbf{0.1174} & \textbf{0.0503} & \textbf{0.0732} & \textbf{0.0156} & \textbf{0.0190} & \textbf{0.1486} & \textbf{0.2048}  & \textbf{0.0524} & \textbf{0.0619} \\
\midrule
\textit{w/o LE}   & 0.0817 & 0.1531 & 0.0873 & 0.1129 & 0.0420 & 0.0671 & 0.0122 & 0.0163 & 0.1355 & 0.1860 & 0.0472 & 0.0556 \\
\textit{w/o SC}  & 0.0835 & 0.1522 & 0.0878 & 0.1120 & 0.0415 & 0.0545 & 0.0113 & 0.0134 & 0.1357 & 0.1884 & 0.0469 & 0.0557 \\
\textit{w/o Gate}  & 0.0786 & 0.1432 & 0.0850 & 0.1063 & 0.0420 & 0.0718 & 0.0131 & 0.0178 & 0.1398 & 0.1932 & 0.0491 & 0.0579 \\
\bottomrule
\end{tabular}%
}
\end{table*}

\noindent\textbf{Removing either restriction degrades every reported metric.}
Table~\ref{tab:multi_dataset_comparison} evaluates the two coordinated restrictions. Removing the user-conditioned latent environment (\textit{w/o LE}) lowers every ranking metric, and Recall at the larger cutoff falls by $5.8\%$ on Synthetic, $8.3\%$ on Meituan, and $9.2\%$ on Yelp. Removing the feature-indexed polynomial mechanism (\textit{w/o SC}) also lowers every metric, and its largest effect is on Meituan, where Recall@100 falls from $0.0732$ to $0.0545$. Meituan is the smallest and sparsest dataset in Table~\ref{tab:Statistics}, where an unrestricted mechanism has more room to overfit. Both variants retain the encoder, decoder, latent capacity, and input separation of the full model, so the differences isolate the contribution of each restriction at fixed capacity and support their coordinated use in CILER.

\noindent\textbf{The auxiliary gate acts as a decoder regularizer outside the identified mechanism.}
Removing the gate (\textit{w/o Gate}) lowers every metric, and the largest relative changes appear on Meituan, where Recall@$50$ falls by $16.5\%$. The paired tests across decoder widths in Appendix~\ref{app:gate_analysis} reach the $0.05$ level only for Yelp mean reciprocal rank (MRR), so the benefit is dataset- and metric-dependent. Neither Theorem~\ref{thm:ident} nor Proposition~\ref{prop:risk} depends on it, and the identified mechanism consists of the user-conditioned environment and the polynomial structural model.

\subsection{Environment marginalization}
\label{sec:marginalization}

\begin{table*}[t]
\centering
\caption{Ablation of environment-marginalized inference on the three datasets.\textit{Single state} replaces the inferred conditional environment distribution with a single environment state, and \textit{w/o posterior inference} removes interaction-conditioned user inference.}
\label{tab:CI}
\begin{tabular*}{\textwidth}{@{\extracolsep{\fill}}l l c c c c@{}}
\toprule
\textbf{Dataset} & \textbf{Variant} & \textbf{R@10} & \textbf{R@20} & \textbf{N@10} & \textbf{N@20} \\
\midrule
\multirow{3}{*}{Synthetic}
& \textit{Single state} & 0.0730 & 0.1436 & 0.0780 & 0.1042 \\
& \textit{w/o posterior inference} & 0.0767 & 0.1429 & 0.0790 & 0.1030 \\
& \textbf{CILER}  & \textbf{0.0882} & \textbf{0.1626} & \textbf{0.0915} & \textbf{0.1174} \\
\midrule
\textbf{Dataset} & \textbf{Variant} & \textbf{R@50} & \textbf{R@100} & \textbf{N@50} & \textbf{N@100} \\
\midrule
\multirow{3}{*}{Meituan}
& \textit{Single state} & 0.0336 & 0.0676 & 0.0094 & 0.0149 \\
& \textit{w/o posterior inference} & 0.0350 & 0.0587 & 0.0104 & 0.0142 \\
& \textbf{CILER}  & \textbf{0.0503} &\textbf{ 0.0732} & \textbf{0.0156} & \textbf{0.0190} \\
\midrule
\multirow{3}{*}{Yelp}
& \textit{Single state} & 0.1379 & 0.1900 & 0.0481 & 0.0568 \\
& \textit{w/o posterior inference} & 0.1382 & 0.1914 & 0.0484 & 0.0572 \\
& \textbf{CILER}  & \textbf{0.1486} & \textbf{0.2048}  & \textbf{0.0524} & \textbf{0.0619} \\
\bottomrule
\end{tabular*}
\end{table*}

\noindent\textbf{Marginalizing the conditional environment outperforms prediction at a single environment state.}
Table~\ref{tab:CI} separates user-specific posterior inference from environment marginalization. Replacing the conditional distribution with one environment state reduces Recall at the smaller cutoff. CILER improves over this variant by $20.8\%$ on Synthetic, $49.7\%$ on Meituan, and $7.8\%$ on Yelp. CILER improves over the variant without interaction-conditioned user inference by $15.0\%$, $43.7\%$, and $7.5\%$ on the same metric. User inference captures variation supported by interaction history. Marginalization integrates environmental variation that remains uncertain after observing user features. The single-state loss is largest on Meituan, consistent with the role of environmental uncertainty in Proposition~\ref{prop:risk}.

\begin{table*}[t]
\centering
\caption{Paired 10-seed diagnostics for environment marginalization on Synthetic. Entries are mean improvements of CILER with 95\% bootstrap confidence intervals. Positive NLL and Brier values denote error reductions.}
\label{tab:controlled_marginalization}
\resizebox{\textwidth}{!}{%
\begin{tabular}{lcccc}
\toprule
\textbf{Reference} & \textbf{R@20} & \textbf{N@20} & \textbf{Positive NLL} & \textbf{Brier score} \\
\midrule
Single environment draw & 0.0045 [0.0021, 0.0073] & 0.0020 [0.0005, 0.0038] & 0.2103 [0.1061, 0.3267] & 0.0025 [0.0017, 0.0034] \\
Context-free environment & 0.0433 [0.0364, 0.0505] & 0.0320 [0.0276, 0.0366] & 0.7886 [0.5130, 1.1228] & 0.0210 [0.0085, 0.0412] \\
\bottomrule
\end{tabular}%
}
\end{table*}

\noindent\textbf{Marginalization improves calibration, as the risk bound predicts.}
Proposition~\ref{prop:risk} concerns logarithmic risk. Table~\ref{tab:controlled_marginalization} tests this claim with paired diagnostics under matched training budgets. Marginalization reduces negative log-likelihood and Brier score and also improves ranking over a single environment draw. Conditioning the environment on user context improves all four diagnostics over a context-free prior. The context-free comparison shows the larger effect because it mismatches the conditional environment across feature values. A single draw from the correct conditional distribution preserves feature dependence and discards uncertainty. Appendix~\ref{app:tce} inspects the same mechanism for individual users.

\subsection{Controlled tests of the conditions}
\label{sec:controlled_claims}
The preceding experiments measure the predictive contribution of the two restrictions and marginalization. Latent recovery requires data with known environment laws, preference mechanisms, and graphs. We evaluate these quantities on controlled data by varying one condition at a time. Each comparison uses ten paired seeds. For each seed, the paired difference is oriented so that a positive value favors the first condition. It is the first value minus the second for higher-is-better metrics and the second value minus the first for lower-is-better metrics. We report both condition means, the mean paired difference, and its $95\%$ percentile-bootstrap confidence interval computed from 10,000 resamples of the ten seed-wise differences. Appendix~\ref{app:synthetic_generators} defines the generators and evaluation protocol. The three blocks of Table~\ref{tab:controlled_structure} test feature-induced variation in its qualitative and rank forms (Assumption~\ref{as:variation}) and the polynomial model class of Theorem~\ref{thm:ident}, and Table~\ref{tab:controlled_family} tests correct specification (Assumption~\ref{as:spec}).

\begin{table*}
\centering
\caption{Controlled comparisons of context variation, sufficient variation, and preference-mechanism recovery over ten paired seeds. The first and second means correspond to the conditions named in each comparison. Positive paired improvements favor the first condition. Brackets report $95\%$ bootstrap confidence intervals.}
\label{tab:controlled_structure}
\resizebox{\textwidth}{!}{%
\begin{tabular}{lccccc}
\toprule
\multirow{2}{*}{\textbf{Comparison}} & \multicolumn{1}{c}{\multirow{2}{*}{\textbf{Metric}}} & \multicolumn{2}{c}{\textbf{Condition means}} & \multicolumn{2}{c}{\textbf{Paired effect}} \\\cmidrule(lr){3-4}\cmidrule(lr){5-6}& & \textbf{First} & \textbf{Second} & \textbf{Improvement} & \textbf{95\% CI} \\
\midrule
\multicolumn{6}{l}{\textit{Context variation}} \\
\hline
Full vs. shuffled & Conditional mean RMSE $\downarrow$ & 0.2717 & 0.4385 & 0.1667 & [0.1501, 0.1833] \\
Full vs. shuffled & Wasserstein-1 $\downarrow$ & 0.3586 & 0.4917 & 0.1331 & [0.0470, 0.2744] \\
Full vs. constant & Conditional mean RMSE $\downarrow$ & 0.2717 & 1.0789 & 0.8072 & [0.4040, 1.3425] \\
Full vs. constant & Wasserstein-1 $\downarrow$ & 0.3586 & 6.5069 & 6.1484 & [2.5801, 10.7647] \\
\hline
\multicolumn{6}{l}{\textit{Sufficient variation}} \\
\hline
Sufficient vs. deficient & Mean correlation coefficient $\uparrow$ & 0.6386 & 0.5462 & 0.0924 & [0.0504, 0.1340] \\
Sufficient vs. deficient & Latent $R^2$ $\uparrow$ & 0.4481 & 0.3562 & 0.0919 & [0.0392, 0.1516] \\
\hline
\multicolumn{6}{l}{\textit{Preference mechanism}} \\
\hline
Polynomial vs. linear & Mechanism MSE $\downarrow$ & 0.5276 & 3.1159 & 2.5884 & [1.3445, 4.2035] \\
Polynomial vs. linear & Coefficient RMSE $\downarrow$ & 0.0843 & 0.4041 & 0.3199 & [0.3103, 0.3280] \\
Polynomial vs. independent & Mechanism MSE $\downarrow$ & 0.5276 & 4.8332 & 4.3057 & [2.5325, 6.4857] \\
Polynomial vs. independent & Coefficient RMSE $\downarrow$ & 0.0843 & 0.2903 & 0.2061 & [0.2008, 0.2122] \\
\bottomrule
\end{tabular}%
}
\end{table*}

\noindent\textbf{The conditional environment requires context to vary.}
The first block of Table~\ref{tab:controlled_structure} tests the role of observed context. Shuffling or fixing the context preserves the model class and removes its association with the environment. Both interventions increase conditional mean error and Wasserstein-1 distance. Constant context produces a substantially larger degradation than shuffled context, exceeding an order of magnitude on Wasserstein-1. This pattern matches Assumption~\ref{as:variation}, which requires feature-induced variation in the natural parameters.

\noindent\textbf{Sufficient variation improves latent recovery.}
The second block tests the rank condition directly. A design whose induced parameter differences span the required dimension yields higher mean correlation and latent $R^2$ than a rank-deficient design. The result links recovery to the variation required by Theorem~\ref{thm:ident}.

\noindent\textbf{The polynomial preference mechanism recovers structure that linear and independent alternatives miss.}
The third block tests the polynomial preference mechanism. Linear and independent alternatives both increase mechanism and coefficient errors. The independent alternative produces the larger mechanism error, and the linear alternative produces the larger coefficient error. The first two blocks vary the conditional environment with a fixed mechanism, and the third varies the mechanism with a fixed environment. Each intervention changes its corresponding recovery metric in the direction predicted by Section~\ref{sec:theory}.

\begin{table*}
\centering
\caption{Improvement of a matched conditional-density model over a Gaussian reference. NLL and CRPS denote negative log-likelihood and continuous ranked probability score. Laplace and Negative Binomial are stress tests outside the regular continuous class covered by the identifiability theorem. Each cell reports the mean paired improvement and its 95\% bootstrap confidence interval.}
\label{tab:controlled_family}
\resizebox{\textwidth}{!}{%
\begin{tabular}{lccc}
\toprule
\textbf{True family} & \textbf{NLL} & \textbf{CRPS} & \textbf{Wasserstein-1} \\
\midrule
Laplace & 0.2939 [0.2762, 0.3122] & 0.0044 [0.0038, 0.0050] & 0.0388 [0.0378, 0.0398] \\
Gamma & 1.1096 [1.0942, 1.1236] & 0.0197 [0.0190, 0.0205] & 0.1716 [0.1696, 0.1742] \\
Beta & 0.3882 [0.3791, 0.3957] & 0.0011 [0.0009, 0.0013] & 0.0147 [0.0145, 0.0149] \\
Negative Binomial & 1.0949 [1.0748, 1.1167] & 0.0644 [0.0630, 0.0657] & 0.3130 [0.3110, 0.3149] \\
\bottomrule
\end{tabular}%
}
\end{table*}

\noindent\textbf{Family specification affects conditional-density recovery.}
Table~\ref{tab:controlled_family} compares a matched conditional-density model with a Gaussian reference under Laplace, Gamma, Beta, and Negative Binomial generators. Gamma and Beta test family specification within the regular continuous class. Laplace and Negative Binomial serve as stress tests outside Theorem~\ref{thm:ident}. The matched model improves negative log-likelihood, CRPS, and Wasserstein-1 distance for every generator. This pattern motivates validation-based family selection within the candidate set. The experiment tests the conditional-density part of Assumption~\ref{as:spec}, which requires the fitted family to match the generating family. The paired improvements quantify the predictive cost of misspecification.

\noindent\textbf{CILER recovers latent variables and graphs across model sizes.}
A complementary study evaluates whole-graph recovery across two to seven latent variables and thirty environments. We compare CILER with variational autoencoder (VAE), $\beta$-VAE, and identifiable VAE (iVAE). The linear generators use Beta or Gamma noise. The nonlinear generator uses Gaussian noise with polynomial relations. Fig.~\ref{fig:synthetic_graphs} shows their shared topology.

\begin{figure}
\centering
\includegraphics[width=1.0\columnwidth]{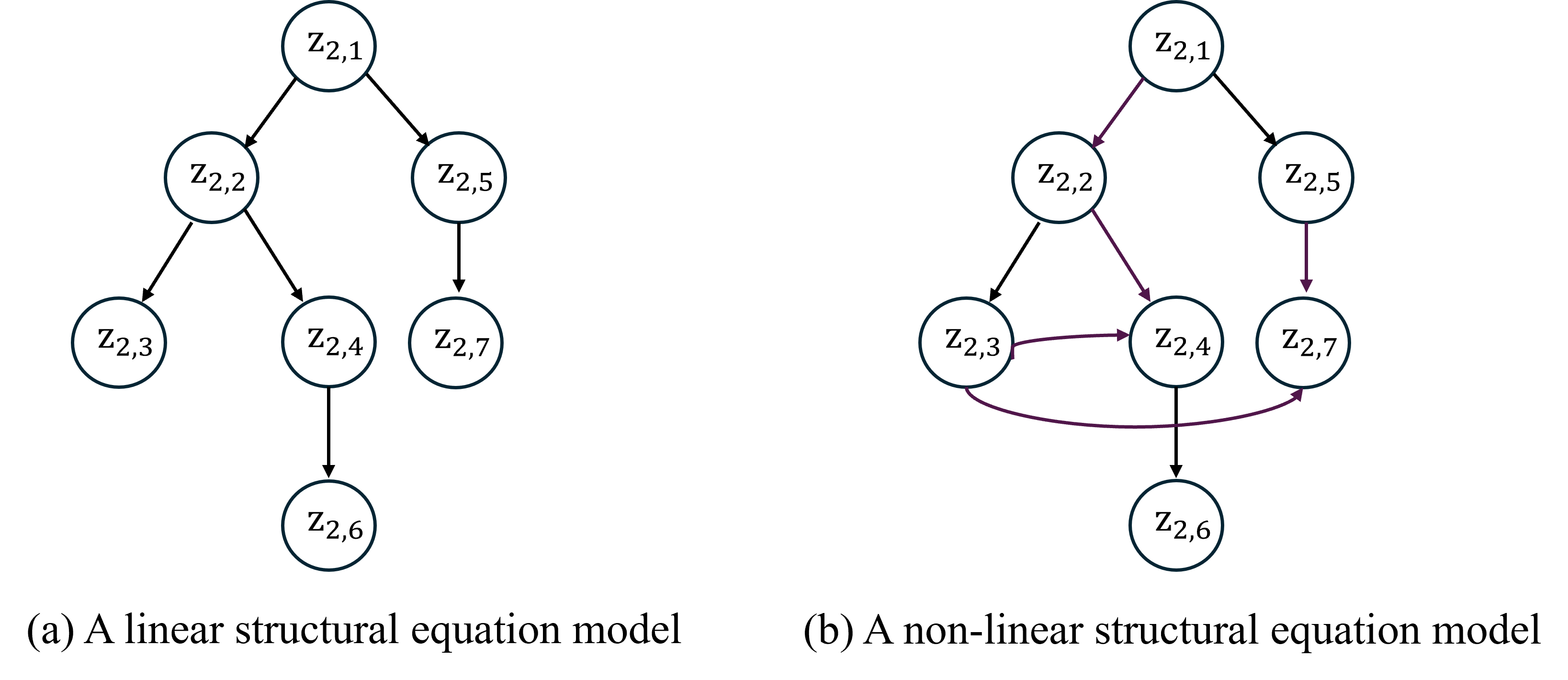}
\caption{Synthetic structures used for controlled recovery. The linear and nonlinear generators share a directed topology. The nonlinear generator replaces selected edges with polynomial relations.}
\label{fig:synthetic_graphs}
\end{figure}

\begin{figure*}
\centering
\includegraphics[width=1.0\textwidth]{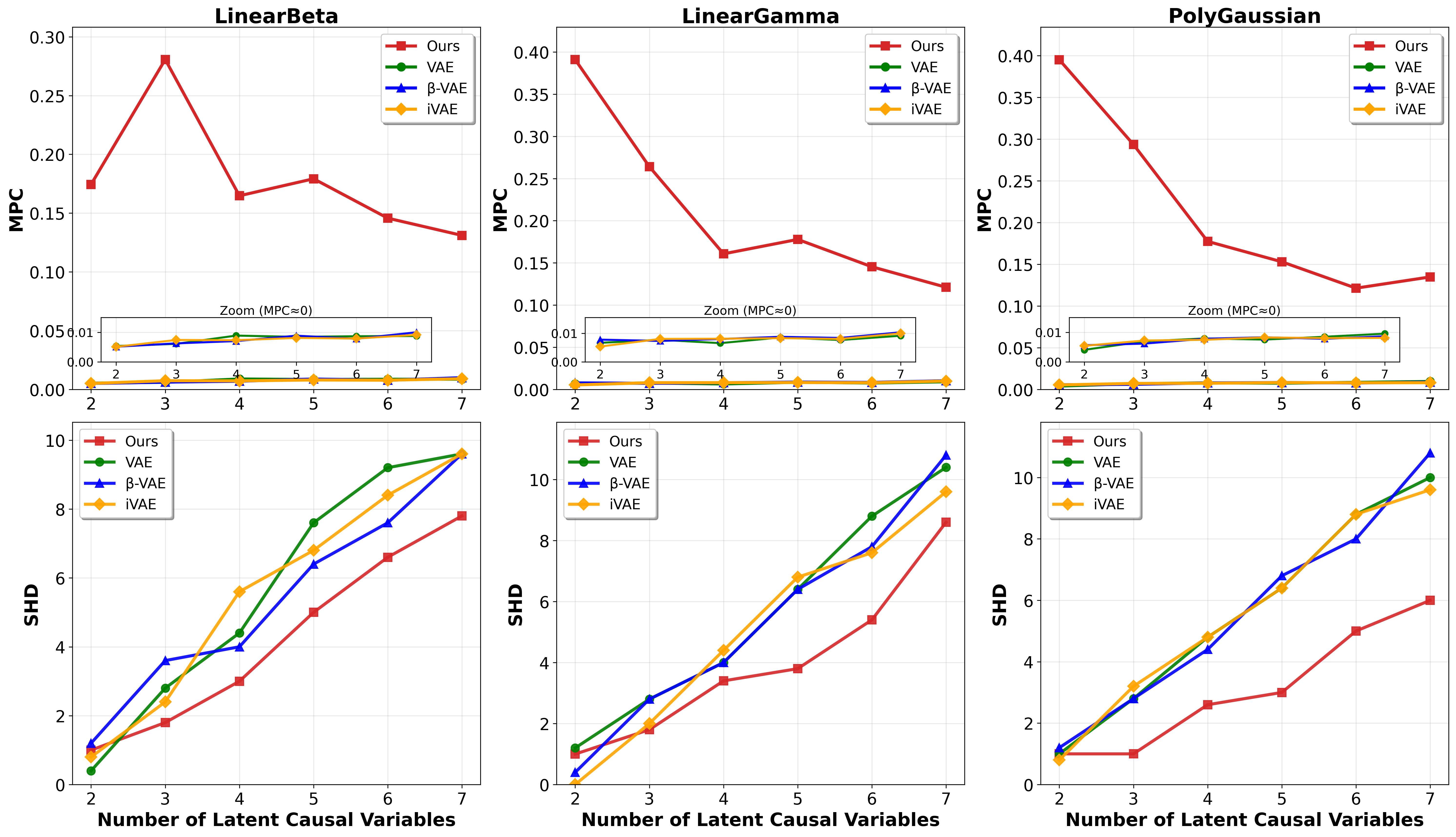}
\caption{Latent-variable and graph recovery as dimensionality increases. Top: mean correlation coefficient, where higher is better. Bottom: structural Hamming distance, where lower is better.}
\label{fig:synthetic_results}
\end{figure*}

CILER achieves the highest mean correlation across all three generators in Fig.~\ref{fig:synthetic_results}. Its structural Hamming distance is lowest or tied for lowest over most dimensions. The largest separation occurs under polynomial relations, which CILER models explicitly. Recovery declines with latent dimensionality for every method. Higher dimensions require more directions of feature-induced variation under Assumption~\ref{as:variation}.

\subsection{Specification and sensitivity}
\label{sec:sensitivity}

\noindent\textbf{The exponential-family member is a dataset-level choice.}
Fig.~\ref{fig:prior_family} compares nine exponential-family members at a common cutoff of $100$, which differs from the per-dataset cutoffs used in the ranking tables. Gaussian, Gamma, Beta, Dirichlet, and Exponential form the regular continuous candidate set. Poisson, Bernoulli, Multinomial, and Negative Binomial provide support-sensitivity diagnostics. Gaussian, Gamma, and Dirichlet perform strongly on Synthetic. Gaussian and Gamma lead Recall on Meituan, and Beta is competitive on NDCG and MRR. Most families perform similarly on Yelp except Negative Binomial. Vertical scales differ across datasets, so each panel supports within-dataset comparison. The identifiability guarantee covers the regular continuous candidates.

\begin{figure*}
\centering
\includegraphics[width=1.0\textwidth]{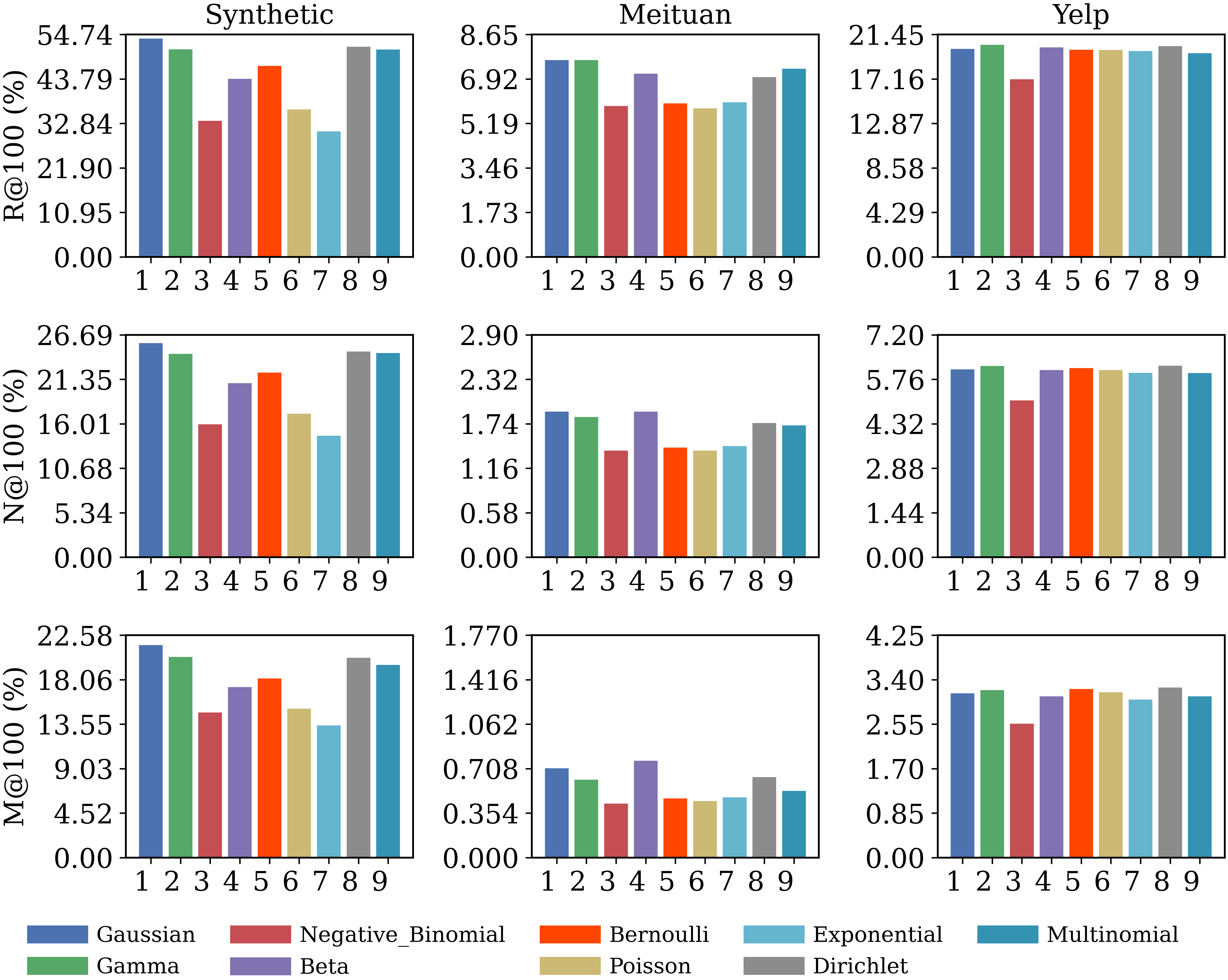}
\caption{Within-dataset comparison of exponential-family choices using Recall@100, NDCG@100, and MRR@100. Family effects are pronounced on Synthetic and Meituan and smaller on Yelp.}
\label{fig:prior_family}
\end{figure*}

\noindent\textbf{Environment dimension has a dataset-dependent operating range.}
Fig.~\ref{fig:parameter_sensitivity} varies the environment dimension $K$ and three loss weights. Synthetic changes modestly across $K$ and stabilizes after a small decline from $K=3$. Meituan peaks near $K=7$ and declines at $K=11$. Yelp remains stable across the tested range. Loss-weight sensitivity follows the same dataset ordering. Synthetic has mild optima near $\alpha=0.5$ and $\beta=10^{-5}$. Meituan responds most strongly to $\lambda$ and $K$, while Yelp remains stable. The Meituan result is consistent with limited temporal-context variation relative to a higher-dimensional environment.

\begin{figure*}
\centering
\includegraphics[width=1.0\textwidth]{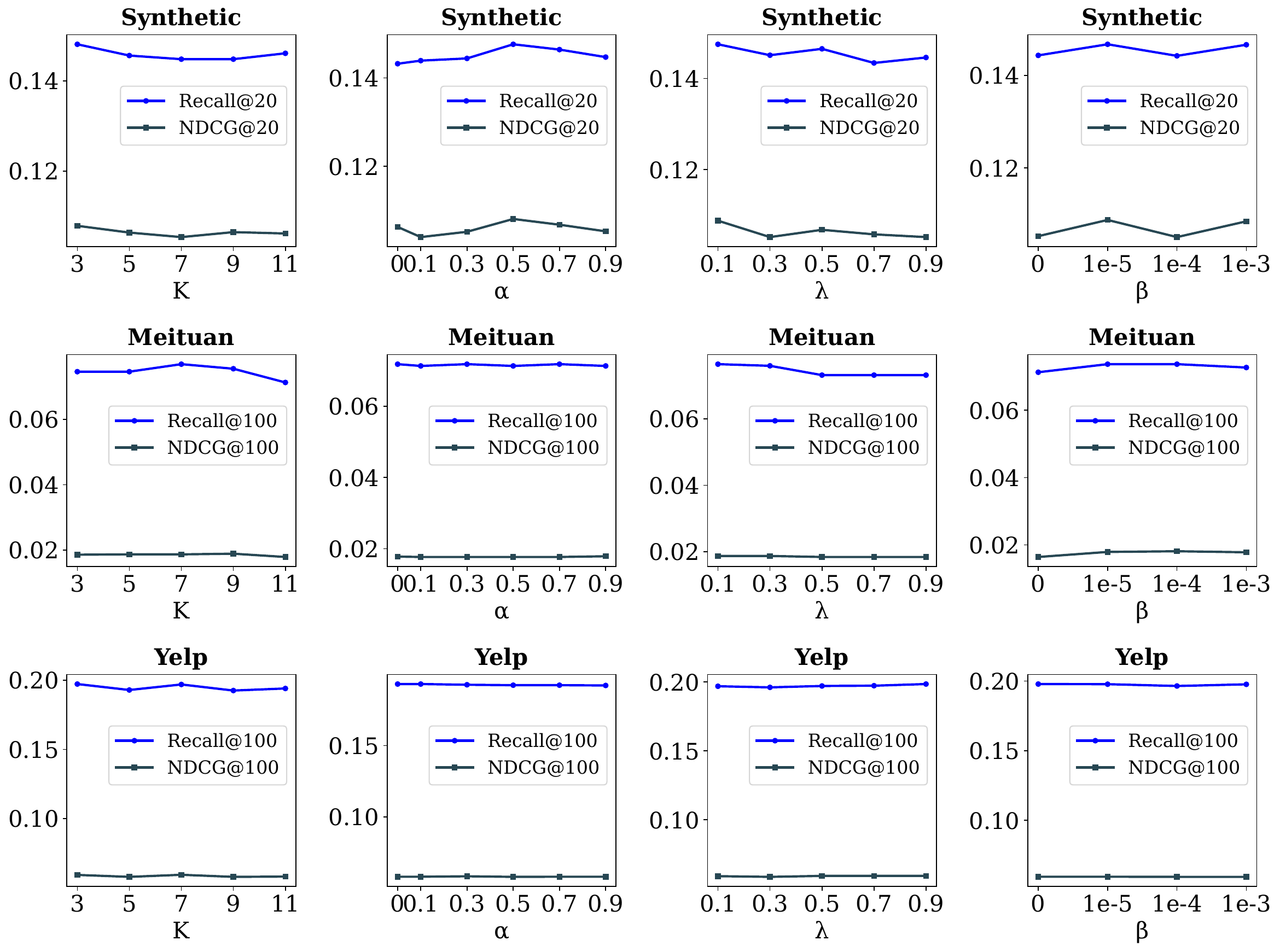}
\caption{Sensitivity to latent dimension $K$, gate weight $\alpha$, structural weight $\lambda$, and weight decay $\beta$. Each row corresponds to one dataset.}
\label{fig:parameter_sensitivity}
\end{figure*}

\begin{figure}
\centering
\includegraphics[width=1.0\columnwidth]{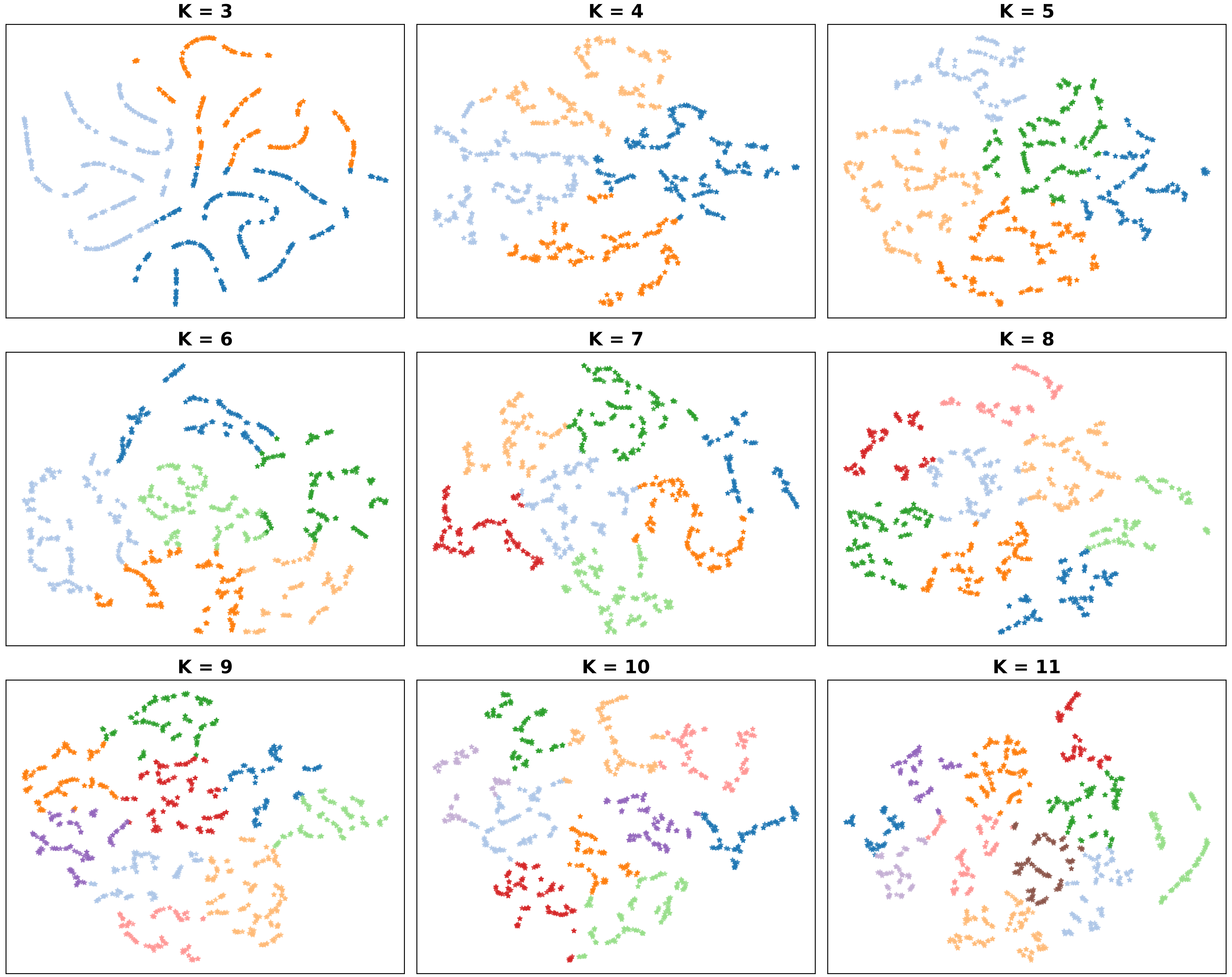}
\caption{t-SNE projections of the Synthetic representations for $K=3,\ldots,11$. Colors distinguish learned component assignments. Components separate at small $K$ and overlap as $K$ grows.}
\label{fig:latent_visualization}
\end{figure}

\noindent\textbf{Larger environment dimensions coincide with component overlap.}
Fig.~\ref{fig:latent_visualization} gives a qualitative view of the dimension sweep. The t-distributed stochastic neighbor embedding (t-SNE \cite{van2008visualizing}) projections of the Synthetic representations form distinct components at small $K$ and show greater overlap as $K$ increases. This pattern agrees with the recovery decline in Fig.~\ref{fig:synthetic_results} and indicates that increasing $K$ does not simply add useful latent structure.

\noindent\textbf{Scope of the empirical gains.}
The gains vary across shift protocols. Yelp shows the smallest OOD improvement and the weakest sensitivity to the environment family, the environment dimension $K$, and the loss weights. This pattern is consistent with less recoverable environment-sensitive variation under its geographical split. The auxiliary gate affects only one dataset--metric pair at the $0.05$ level, so it remains an implementation component rather than a core mechanism. Recovery declines with latent dimensionality for every method in Fig.~\ref{fig:synthetic_results}. The identifiability conditions are most informative when feature-induced variation is sufficient relative to the latent dimension.

\subsection{Complexity and scalability}
\label{sec:complexity}

\noindent\textbf{The restrictions add bounded overhead and no test-time updates.}
The dense encoder--decoder path costs $\mathcal{O}(N\sum_l d_l d_{l-1})$ per epoch. Conditional environment inference and the polynomial mechanism add $\mathcal{O}(N(dK+K^2))$, which remains below the dense cost when $K\ll d$. At deployment, marginalization adds $S$ batched evaluations of the sensitive branch. CILER performs no gradient updates at test time. Table~\ref{tab:complexity} separates shared neural computation from method-specific overhead. CILER has the same asymptotic class as MultiVAE and COR. DR-GNN adds graph propagation, and DT3OR adds test-time adaptation. Table~\ref{tab:complexity} compares asymptotic orders of growth and reports no measured runtime.

\begin{table}[t]
\centering
\caption{Theoretical complexity comparison of key methods.}
\label{tab:complexity}
\resizebox{\columnwidth}{!}{
\small
\begin{tabular}{lcc}
\toprule
\textbf{Method} & \textbf{Training Complexity} & \textbf{Additional Overhead} \\
\midrule
FM    & $\mathcal{O}(Nd^2)$ & None \\
MultiVAE & $\mathcal{O}(N\sum_l d_l d_{l-1})$ & None \\
MacridVAE & $\mathcal{O}(N\sum_l d_l d_{l-1})$ & $\mathcal{O}(NK^2)$ \\
COR   & $\mathcal{O}(N\sum_l d_l d_{l-1})$ & $\mathcal{O}(NK^2)$ \\
DR-GNN & $\mathcal{O}(|\mathcal{E}|N\sum_l d_l d_{l-1})$ & Graph propagation \\
DT3OR  & $\mathcal{O}(2N\sum_l d_l d_{l-1})$ & Test-time adaptation \\
\midrule
\textbf{CILER} & $\mathcal{O}(N\sum_l d_l d_{l-1})$ & $\mathcal{O}(N(dK + K^2))$ \\
\bottomrule
\end{tabular}%
}
\end{table}

\noindent\textbf{Additional results.}
The appendices provide technical and experimental details. Appendix~\ref{app:variational_objective} derives the training objective from the evidence lower bound. Appendix~\ref{app:structural_details} defines the polynomial basis and reference-order preference mechanism. Appendix~\ref{app:experimental_details} describes the baselines and search spaces. Appendix~\ref{app:synthetic_generators} specifies the controlled generators and the paired evaluation protocol. Appendix~\ref{app:gate_analysis} reports the architecture study and paired tests for the auxiliary gate.

\section{Related work}
\label{sec:related}
Our work relates to out-of-distribution recommendation, environment inference, and identifiable latent representation learning. These areas define the prediction setting, the unobserved source of shift, and the structural conditions used by CILER.\par

\noindent\textbf{Out-of-distribution recommendation.}
Out-of-distribution recommenders aim to preserve ranking performance when interaction distributions change. Early methods learn invariant preferences or separate stable preference from environment-sensitive preference~\citep{zhang2023reformulating,du2022invariant,wang2022invariant}. Recent methods address specified shifts through distributionally robust objectives, popularity-aware modeling, causal graph learning, or test-time adaptation~\citep{wang2024distributionally,zhang2024popshift,wang2024biasagnostic,zhao2025graph,yang2025dual}. Recent recommendation studies further show that discarding variant preference can reduce generalization, while sequential models construct auxiliary environments to learn stable preference under shift~\citep{bai2025invariantdebias,liao2025mitigating}. Cross-domain methods transfer invariances from labeled domains or adapt with deployment interactions~\citep{zhang2023connecting,zhang2023hierarchical,yoo2025generalizable}. These methods improve OOD ranking, yet their objectives do not determine the distribution of an unobserved environment or its effect on preference. CILER constrains both components and marginalizes item probabilities over the inferred environment distribution at deployment.\par

\noindent\textbf{Environment inference.}
Invariant learning typically assumes observed environment labels, and its guarantees depend on the supplied partition~\citep{arjovsky2019invariant,rojas2018invariant,ahuja2021invariance,kamath2021does,huh2022missing}. When labels are unavailable, EIIL and related methods infer discrete partitions that reveal invariance violations~\citep{creager2021environment,zhu2024clustering,yuan2023environment}. TIVA estimates a partition from observed attributes~\citep{tan2023provably}. Recent methods infer task-specific partitions from temporal structure or feature decorrelation~\citep{liu2024foil,liao2026decorr}. Continuously indexed adaptation replaces discrete domains with a domain index~\citep{wang2020continuously}. These methods use the inferred environment to train an invariant predictor. CILER learns a conditional latent-environment distribution and marginalizes item probabilities over this distribution at deployment.\par

\noindent\textbf{Identifiable latent representation learning.}
Latent representations are not identifiable from observations alone without additional structure~\citep{locatello2019challenging}. Nonlinear ICA introduces observed auxiliary variables that modulate exponential-family latent distributions~\citep{hyvarinen2019nonlinear,khemakhem2020variational}. Recent causal representation results identify latent variables from unknown or soft interventions across multiple environments~\citep{vonkugelgen2023nonparametric,zhang2023softinterventions,buchholz2023linearcausal}. Score-based methods recover latent causal models from changes in density scores~\citep{varici2025scorebased}. General-environment and weight-variant models use sufficient mechanism variation across observed conditions~\citep{ng2025general,liu2026weightvariant}. Polynomial causal models identify nonlinear mechanisms when their coefficients vary sufficiently across conditions~\citep{liu2024identifiable}. CILER brings these conditions to OOD recommendation through a user-conditioned exponential-family environment and a feature-indexed polynomial preference mechanism.\par

\section{Conclusion}
\label{sec:conclusion}
In this work, we presented Conditionally Identifiable Latent-Environment Recommendation (CILER) for OOD recommendation under latent environment shifts. CILER restricts the latent-environment distribution and the environment-sensitive preference mechanism, then marginalizes item probabilities over the inferred environment distribution at deployment. Theoretically, we established conditional identifiability of the environment-sensitive representation under the stated conditions and bounded excess deployment log-risk by environment-inference error. Empirically, CILER ranks first on all twelve OOD ranking metrics across three datasets while retaining competitive IID accuracy. Controlled studies and capacity-matched ablations further connect recovery and ranking performance to the stated conditions and the two restrictions. CILER links latent-environment modeling to OOD recommendation through conditional identifiability and deployment risk.

\section*{CRediT authorship contribution statement}
\textbf{Qianqian Wang:} Conceptualization, Methodology, Software, Validation, Formal analysis, Investigation, Writing -- original draft. 
\textbf{Wenwu Gong:} Writing – review \& editing,
Methodology, Conceptualization.
\textbf{Yunshan Li:} Software, Resources.
\textbf{Zhenqing Wu:} Software, Resources.
\textbf{Ruili Wang:} Writing – review \& editing.
\textbf{Lili Yang:} Supervision,
Resources, Writing – review \& editing.

\section*{Declaration of competing interest}
The authors declare that they have no known competing financial interests or personal relationships that could have appeared to influence the work reported in this paper.

\section*{Acknowledgements}
This research was funded by the SUSTech Presidential Postdoctoral Fellowship, the China Postdoctoral Science Foundation (Grant No. 2025M773057), Shenzhen Science and Technology Program (Grant No. ZDSYS20210623092007023), and the Shenzhen Key Laboratory of Safety and Security for Next Generation of Industrial Internet, Southern University of Science and Technology.

\section*{Data availability}
All three evaluation protocols are public. The Synthetic, Meituan, and Yelp datasets and their out-of-distribution splits are those released with COR~\citep{wang2022causal}. Our code will be publicly available upon acceptance.

\section*{Declaration of generative AI use}
The authors used generative AI tools only for language polishing and readability improvement.  All technical ideas, experiments, analyses, and conclusions were developed and verified by the authors, who take full responsibility for the content of this paper.

\appendix
\numberwithin{equation}{section}
\numberwithin{figure}{section}
\numberwithin{table}{section}
\section{Proofs of theoretical results}
\label{app:proofs}

\subsection{Proof of Proposition~\ref{prop:entangle}}
\label{app:proof_entangle}
\begin{proof}
(i) Let $\phi$ be a diffeomorphism of $\mathcal{E}$ and set $\widetilde{\boldsymbol{\varepsilon}}=\phi(\boldsymbol{\varepsilon})$. Its conditional law is the pushforward $\phi_{\#}p(\cdot\mid\boldsymbol{x}_1)$, which remains smooth under the change-of-variables formula. Define $\widetilde{g}(\boldsymbol{x}_2,\widetilde{\boldsymbol{\varepsilon}})=g(\boldsymbol{x}_2,\phi^{-1}(\widetilde{\boldsymbol{\varepsilon}}))$. Then $\widetilde{g}(\boldsymbol{x}_2,\widetilde{\boldsymbol{\varepsilon}})=g(\boldsymbol{x}_2,\boldsymbol{\varepsilon})$ almost surely. The joint representation and decoder outputs remain unchanged. Thus $\mathcal{M}_{\phi}\in\mathbb{E}_{\mathcal{F}_{\mathrm{all}}}(\mathcal{M})$ with latent environment $\phi(\boldsymbol{\varepsilon})$. The interaction distribution identifies the environment only up to a smooth invertible reparameterization.

(ii) Let $\boldsymbol{A}$ be invertible on the representation space. Set $\widetilde{\boldsymbol{z}}=\boldsymbol{A}\boldsymbol{z}$ and $\widetilde{h}=h\circ\boldsymbol{A}^{-1}$. The identity $\widetilde{h}(\widetilde{\boldsymbol{z}})=h(\boldsymbol{z})$ preserves the item scores and $p(\boldsymbol{y}\mid\boldsymbol{x}_1)$. The transformed model belongs to $\mathbb{E}_{\mathcal{F}_{\mathrm{all}}}(\mathcal{M})$. Write $\boldsymbol{A}$ in blocks conformal with $[\boldsymbol{z}_1,\boldsymbol{z}_2]$. The first $\dim(\boldsymbol{z}_1)$ transformed coordinates equal $\boldsymbol{A}_{11}\boldsymbol{z}_1+\boldsymbol{A}_{12}\boldsymbol{z}_2$. When $\boldsymbol{A}_{12}\neq\boldsymbol{0}$ and $\boldsymbol{A}_{12}\boldsymbol{z}_2$ is not almost surely constant in $\boldsymbol{\varepsilon}$, these coordinates become environment-dependent. The transformed model still lies in $\mathcal{F}_{\mathrm{all}}$, since the class places no restriction on which coordinates form the stable block. The same observational distribution therefore supports different stable--sensitive decompositions.
\end{proof}

\subsection{Proof of Theorem~\ref{thm:ident}}
\label{app:proof_ident}
\begin{proof}
The argument exhibits $\mathcal{F}_{\mathrm{CILER}}$ as an instance of the latent polynomial causal model of~\citet{liu2024identifiable} and verifies their conditions.

\emph{Correspondence.} Take the observed user features as the auxiliary variable, $\boldsymbol{u}:=\boldsymbol{x}_1$. By Assumption~\ref{as:spec} the conditional environment is $p_{\theta_{\varepsilon}}(\boldsymbol{\varepsilon}\mid\boldsymbol{u})\propto b(\boldsymbol{\varepsilon})\exp(\boldsymbol{\eta}(\boldsymbol{u})^{\top}\boldsymbol{T}(\boldsymbol{\varepsilon}))$, which is the auxiliary-conditioned exponential-family noise distribution their model assumes, and the environment-sensitive variables satisfy Eq.~\eqref{causalZ2} with polynomial $g_i$ of degree at most $p$ and coefficients $\boldsymbol{\Lambda}(\boldsymbol{u})$, which is their latent polynomial structural equation model with auxiliary-dependent coefficients. The lower-triangular structure of $\boldsymbol{\Lambda}(\boldsymbol{u})$ enforces acyclicity under the reference ordering in Appendix~\ref{app:structural_details}.

\emph{Verification.} Assumption~\ref{as:variation} supplies $mk+1$ auxiliary-variable values whose induced parameters satisfy the sufficient-change condition. Assumption~\ref{as:decoder} treats $f_{\theta_3}$ as a smooth injective mixing function and, through its second part, makes the item logits recoverable from $p(\boldsymbol{y}\mid\boldsymbol{x}_1)$. That second part is an assumption rather than a consequence. Recommendation observes multinomial counts, and the multinomial likelihood determines only the softmax probabilities, hence the logits up to an additive constant per user. Recovering the mixing law from a mixture of multinomials needs the further condition stated in Assumption~\ref{as:decoder}. Remark~\ref{rem:scope} records this requirement. Assumption~\ref{as:shift} keeps the conditional environment, mechanism, and decoder invariant across training and deployment. Identification on the training distribution then transfers within the deployment support.

\emph{Conclusion.} Applying their identifiability theorem, any two models in $\mathcal{F}_{\mathrm{CILER}}$ that induce the same observational distribution have environment-sensitive variables that agree up to permutation and component-wise invertible transformation, and directed relations that agree up to the same equivalence. By Definition~\ref{def:oec} this means that $\mathbb{E}_{\mathcal{F}_{\mathrm{CILER}}}(\mathcal{M})$ is contained in the $\approx$-class of $\mathcal{M}$, which is the constraint in Eq.~\eqref{eq:cirr}. We claim the inclusion only. A component-wise invertible transformation of the environment-sensitive coordinates need not preserve the polynomial form of Eq.~\eqref{causalZ2}, so the $\approx$-class of $\mathcal{M}$ need not lie inside $\mathcal{F}_{\mathrm{CILER}}$. The inclusion is what the CI-RR constraint requires.
\end{proof}

\subsection{Proof of Proposition~\ref{prop:risk}}
\label{app:proof_risk}
\begin{proof}
Fix $\boldsymbol{x}_1$ and instantiate Eq.~\eqref{eq:robust_def} with the logarithmic loss $\ell(p;\boldsymbol{y})=-\log p(\boldsymbol{y}\mid\boldsymbol{x}_1)$. Under Assumption~\ref{as:shift} the interactions of a user with features $\boldsymbol{x}_1$ are drawn from $p^{*}(\boldsymbol{y}\mid\boldsymbol{x}_1)=\int p(\boldsymbol{y}\mid\boldsymbol{x}_1,\boldsymbol{\varepsilon})p^{*}(\boldsymbol{\varepsilon}\mid\boldsymbol{x}_1)\,d\boldsymbol{\varepsilon}$. This gives
\begin{equation*}
\begin{aligned}
\mathcal{R}_{\mathrm{te}}(p_{q})-\mathcal{R}_{\mathrm{te}}(p^{*})
&=\mathbb{E}_{\boldsymbol{x}_1}\mathbb{E}_{\boldsymbol{y}\sim p^{*}(\cdot\mid\boldsymbol{x}_1)}
\left[\log\frac{p^{*}(\boldsymbol{y}\mid\boldsymbol{x}_1)}{p_{q}(\boldsymbol{y}\mid\boldsymbol{x}_1)}\right]\\
&=\mathbb{E}_{\boldsymbol{x}_1}\mathrm{KL}\!\left(p^{*}(\boldsymbol{y}\mid\boldsymbol{x}_1)\,\|\,p_{q}(\boldsymbol{y}\mid\boldsymbol{x}_1)\right),
\end{aligned}
\end{equation*}
which is the first equality of Eq.~\eqref{eq:risk_transfer}. For the inequality, consider the two joint laws of $(\boldsymbol{\varepsilon},\boldsymbol{y})$ given $\boldsymbol{x}_1$,
$P^{*}=p^{*}(\boldsymbol{\varepsilon}\mid\boldsymbol{x}_1)\,p(\boldsymbol{y}\mid\boldsymbol{x}_1,\boldsymbol{\varepsilon})$ and
$Q=q(\boldsymbol{\varepsilon}\mid\boldsymbol{x}_1)\,p(\boldsymbol{y}\mid\boldsymbol{x}_1,\boldsymbol{\varepsilon})$,
whose $\boldsymbol{y}$-marginals are $p^{*}(\boldsymbol{y}\mid\boldsymbol{x}_1)$ and $p_{q}(\boldsymbol{y}\mid\boldsymbol{x}_1)$. Because both laws use the same channel $p(\boldsymbol{y}\mid\boldsymbol{x}_1,\boldsymbol{\varepsilon})$, the chain rule for the Kullback--Leibler divergence gives
\begin{equation*}
\begin{aligned}
\mathrm{KL}(P^{*}\,\|\,Q)
={}&\mathrm{KL}\!\left(p^{*}(\boldsymbol{\varepsilon}\mid\boldsymbol{x}_1)\,\|\,q(\boldsymbol{\varepsilon}\mid\boldsymbol{x}_1)\right)\\
&+\mathbb{E}_{\boldsymbol{\varepsilon}\sim p^{*}}\!\left[
\mathrm{KL}\!\left(p(\boldsymbol{y}\mid\boldsymbol{x}_1,\boldsymbol{\varepsilon})\,\|\,
p(\boldsymbol{y}\mid\boldsymbol{x}_1,\boldsymbol{\varepsilon})\right)\right]\\
={}&\mathrm{KL}\!\left(p^{*}(\boldsymbol{\varepsilon}\mid\boldsymbol{x}_1)\,\|\,q(\boldsymbol{\varepsilon}\mid\boldsymbol{x}_1)\right),
\end{aligned}
\end{equation*}
since the second term vanishes when both laws use the same channel.
Marginalizing $\boldsymbol{\varepsilon}$ is a deterministic map applied to both laws, so the data-processing inequality yields
$\mathrm{KL}(p^{*}(\boldsymbol{y}\mid\boldsymbol{x}_1)\,\|\,p_{q}(\boldsymbol{y}\mid\boldsymbol{x}_1))\leq\mathrm{KL}(P^{*}\,\|\,Q)$.
Taking the expectation over $\boldsymbol{x}_1\sim P_{\mathrm{te}}$ gives Eq.~\eqref{eq:risk_transfer}. Under Assumption~\ref{as:shift} the density ratio obeys $dP_{\mathrm{te}}/dP_{\mathrm{tr}}\le B$, so the right-hand side is at most $B\,\mathbb{E}_{\boldsymbol{x}_1\sim P_{\mathrm{tr}}}\mathrm{KL}(p^{*}(\boldsymbol{\varepsilon}\mid\boldsymbol{x}_1)\,\|\,q(\boldsymbol{\varepsilon}\mid\boldsymbol{x}_1))$, which is a training-time quantity.

For the final claim, let $q(\cdot\mid\boldsymbol{x}_1)=\delta_{\widehat{\boldsymbol{\varepsilon}}(\boldsymbol{x}_1)}$. Whenever $p^{*}(\cdot\mid\boldsymbol{x}_1)$ is not the point mass at $\widehat{\boldsymbol{\varepsilon}}(\boldsymbol{x}_1)$, it is not absolutely continuous with respect to $q(\cdot\mid\boldsymbol{x}_1)$ and the divergence $\mathrm{KL}(p^{*}\,\|\,q)$ is infinite. If this occurs on a set of features of positive $P_{\mathrm{te}}$-probability, the right-hand side of Eq.~\eqref{eq:risk_transfer} is infinite and the bound carries no information.
\end{proof}

\section{Variational objective derivation}\label{app:variational_objective}
Formally, we optimize model parameters $\{{\theta_1}, {\theta_2}, {\theta_3}, \theta_{\varepsilon}\}$ by maximizing the marginal likelihood of the interaction history $\boldsymbol{y}$ conditioned on observed features $\boldsymbol{x}_1$:
 \begin{equation}
  \log p(\boldsymbol{y}\mid \boldsymbol{x}_1)
  = \log \iint p(\boldsymbol{y},\boldsymbol{x}_2,\boldsymbol{\varepsilon}\mid \boldsymbol{x}_1)\,
  \mathrm{d}\boldsymbol{x}_2\, \mathrm{d}\boldsymbol{\varepsilon}.
  \label{MAP}
  \end{equation}

  The generative model assumes $\boldsymbol{x}_2$ is independent of $\boldsymbol{x}_1$ and uses an unconditional prior $p(\boldsymbol{x}_2)$. The joint factorizes as
  $p(\boldsymbol{y},\boldsymbol{x}_2,\boldsymbol{\varepsilon}\mid\boldsymbol{x}_1)
   = p(\boldsymbol{y}\mid\boldsymbol{x}_1,\boldsymbol{x}_2,\boldsymbol{\varepsilon})\,
     p(\boldsymbol{x}_2)\,p(\boldsymbol{\varepsilon}\mid\boldsymbol{x}_1)$.
  We introduce a mean-field variational family
  $q(\boldsymbol{x}_2,\boldsymbol{\varepsilon}\mid\cdot)
    \coloneqq q(\boldsymbol{x}_2\mid\cdot)\,q(\boldsymbol{\varepsilon}\mid\cdot)$,
  where the two factors are fitted independently.
  Multiplying and dividing the integrand by
  $q(\boldsymbol{x}_2\mid\cdot)\,q(\boldsymbol{\varepsilon}\mid\cdot)$
  and applying Jensen's inequality yields the evidence lower bound (ELBO):
  \begin{equation}
  \begin{aligned}
  \log p(\boldsymbol{y}\mid\boldsymbol{x}_1)
  \geq{}& \mathbb{E}_{q(\boldsymbol{x}_2,\boldsymbol{\varepsilon}\mid\cdot)}\!\Bigg[
  \log p(\boldsymbol{y}\mid\boldsymbol{x}_1,\boldsymbol{x}_2,\boldsymbol{\varepsilon}) \\
  &+ \log\frac{p(\boldsymbol{x}_2)}{q(\boldsymbol{x}_2\mid\cdot)}
  + \log\frac{p(\boldsymbol{\varepsilon}\mid\boldsymbol{x}_1)}
  {q(\boldsymbol{\varepsilon}\mid\cdot)}\Bigg] \\
  \triangleq{}& \mathrm{ELBO}.
  \end{aligned}
  \label{ELBO}
  \end{equation}
  
  The gap equals $\mathrm{KL}\bigl(q(\boldsymbol{x}_2,\boldsymbol{\varepsilon}\mid\cdot)
  \;\|\; p(\boldsymbol{x}_2,\boldsymbol{\varepsilon}\mid\boldsymbol{y},\boldsymbol{x}_1)\bigr)\geq 0$,
  so the bound is tight if and only if $q=p$.
  The conditional independence $\boldsymbol{\varepsilon}\perp\!\!\!\perp\boldsymbol{x}_2\mid\boldsymbol{x}_1$ gives
  $p(\boldsymbol{x}_2,\boldsymbol{\varepsilon}\mid\boldsymbol{x}_1)
=p(\boldsymbol{x}_2)\,p(\boldsymbol{\varepsilon}\mid\boldsymbol{x}_1)$. The evidence lower bound (ELBO) uses the user-conditioned environmental prior $p(\boldsymbol{\varepsilon}\mid\boldsymbol{x}_1)$. Expanding the log ratio and using the mean-field factorisation, we obtain the
  equivalent form
  \begin{align}
  \mathrm{ELBO} ={}&
  \mathbb{E}_{q(\boldsymbol{x}_2,\boldsymbol{\varepsilon}\mid\cdot)}
  \bigl[\log p(\boldsymbol{y}\mid\boldsymbol{x}_1,\boldsymbol{x}_2,\boldsymbol{\varepsilon})\bigr] \\
  &- \mathbb{E}_{q(\boldsymbol{\varepsilon}\mid\cdot)}
  \bigl[\mathrm{KL}\bigl(q(\boldsymbol{x}_2\mid\cdot)\,\|\,p(\boldsymbol{x}_2)\bigr)\bigr] \\
  &- \mathbb{E}_{q(\boldsymbol{x}_2\mid\cdot)}
  \bigl[\mathrm{KL}\bigl(q(\boldsymbol{\varepsilon}\mid\cdot)\,\|\,
  p(\boldsymbol{\varepsilon}\mid\boldsymbol{x}_1)\bigr)\bigr].
  \label{ELBO_expand}
  \end{align}
  The first KL term is independent of $\boldsymbol{\varepsilon}$, so its outer expectation can be dropped, yielding
  $\mathrm{KL}(q(\boldsymbol{x}_2\mid\cdot)\|p(\boldsymbol{x}_2))$ directly.
  A hyperparameter $\lambda\geq 0$ balances the KL regularization against
  reconstruction quality~\citep{wang2023causal}, giving the single joint training
  objective
  \begin{equation}
  \begin{aligned}
  \min_{\Theta}\Bigl\{&
  -\mathbb{E}_{q(\boldsymbol{x}_2,\boldsymbol{\varepsilon}\mid\cdot)}
  \bigl[\log p(\boldsymbol{y}\mid\boldsymbol{x}_1,\boldsymbol{x}_2,
  \boldsymbol{\varepsilon})\bigr] \\
  &+\lambda\,\mathrm{KL}\bigl(q(\boldsymbol{x}_2\mid\cdot)\,\|\,
  p(\boldsymbol{x}_2)\bigr) \\
  &+\lambda\,\mathrm{KL}\bigl(q(\boldsymbol{\varepsilon}\mid\cdot)\,\|\,
  p(\boldsymbol{\varepsilon}\mid\boldsymbol{x}_1)\bigr)\Bigr\}.
  \end{aligned}
  \label{ELBO_loss}
  \end{equation}

\section{Structural parameterization}\label{app:structural_details}
CILER uses a degree-$p$ polynomial basis with distinct monomials:
\begin{equation}
\boldsymbol{\phi}(\boldsymbol{z})
=\left[\boldsymbol{z},\boldsymbol{z}\,\bar{\otimes}\,\boldsymbol{z},\ldots,
\underbrace{\boldsymbol{z}\,\bar{\otimes}\cdots\bar{\otimes}\boldsymbol{z}}_{p}\right].
\end{equation}
We set $p=2$, which gives $k+\binom{k+1}{2}$ features per node. A lower-triangular coefficient matrix imposes the reference ordering $z_{2,1}\succ z_{2,2}\succ\cdots\succ z_{2,k}$ on coordinate indices by setting $\Lambda_{ij}(\boldsymbol{x}_1)=0$ for $j\geq i$. This parameterization guarantees acyclicity. The ordering constrains index positions alone, and the identifiability result assigns latent factors to those positions. The $\ell_1$ term in Eq.~\eqref{causalloss} selects active edges from the full lower-triangular graph.

The SEM-implied conditional distribution is
\begin{equation}
\begin{aligned}
q(\boldsymbol{z}_{2,i}'\mid\operatorname{pa}(\boldsymbol{z}_{2,i}),
\boldsymbol{\varepsilon}_i)
={}&h(\boldsymbol{z}_{2,i})\exp\!\Bigl\{
\eta_i^\top T(\boldsymbol{z}_{2,i})-A(\eta_i)\Bigr\},
\end{aligned}
\end{equation}
where $\eta_i=\eta_i(\operatorname{pa}(\boldsymbol{z}_{2,i}),\boldsymbol{\varepsilon}_i)$. The SEM-fit term in Eq.~\eqref{causalloss} is
\begin{equation}
\mathcal{L}_{\mathrm{SEM}}
=\sum_{i=1}^{k}\mathrm{KL}\!\left(
q(\boldsymbol{z}_{2,i}'\mid\operatorname{pa}(\boldsymbol{z}_{2,i}),\boldsymbol{\varepsilon}_i)
\,\|\,
p(\boldsymbol{z}_{2,i}'\mid\operatorname{pa}(\boldsymbol{z}_{2,i}),\boldsymbol{\varepsilon}_i)
\right).
\end{equation}

\section{Experimental details}\label{app:experimental_details}
\subsection{Baseline descriptions}

We compare CILER against 13 baseline models spanning all major paradigms relevant to OOD recommendation:
\begin{itemize}
    \item \textbf{FM}~\citep{rendle2010factorization}: A classical feature-based model that captures pairwise feature interactions via factorized parameterization.
    \item \textbf{NFM}~\citep{he2017neural}: An extension of FM using a neural architecture for higher-order feature interaction modeling.
    \item \textbf{MultiVAE}~\citep{liang2018variational}: A variational autoencoder-based collaborative filtering model for learning latent representations from implicit feedback.
    \item \textbf{MacridVAE}~\citep{ma2019learning}: A disentangled VAE-based model decomposing user preferences into macro- and micro-level factors.
    \item \textbf{MacridVAE+FM}: A hybrid baseline combining MacridVAE's disentangled representations with FM's explicit feature modeling.
    \item \textbf{CausPref}~\citep{he2022causpref}: A causal preference learning framework incorporating a causal graph structure with anti-bias negative sampling.
    \item \textbf{COR}~\citep{wang2022causal}: A causal recommendation method modeling latent confounders and performing counterfactual reasoning for invariant preference recovery.
    \item \textbf{InvCF}~\citep{zhang2023invariant}: An invariant collaborative filtering model disentangling user preferences from spurious correlations such as item popularity.
    \item \textbf{CDR}~\citep{wang2023causal}: A causal representation learning framework for sequence recommendation using variational inference and causal hard intervention.
    \item \textbf{PopGo}~\citep{zhang2024popshift}: A debiasing model learning a shortcut detector to quantify popularity-induced bias.
    \item \textbf{DR-GNN}~\citep{wang2024distributionally}: A distributionally robust optimization-based graph neural network (GNN) minimizing worst-case risk over environment partitions.
    \item \textbf{CausalDiffRec}~\citep{zhao2025graph}: A causal diffusion framework leveraging backdoor adjustment for invariant graph representation learning.
    \item \textbf{DT3OR}~\citep{yang2025dual}: A dual test-time training framework adapting recommendation models via self-refinement and contrastive learning during inference.
\end{itemize}

\subsection{Implementation details}

We implement CILER in PyTorch and optimize it with Adam. Hyperparameter search uses 300 epochs on Synthetic, 300 on Meituan, and 30 on Yelp. Final evaluation uses 10 random seeds and 1,000, 300, and 200 epochs on the three datasets. The learning rate is $0.001$ for Synthetic and Meituan and $0.00075$ for Yelp. Grid search covers weight decay in $\{0,10^{-5},\ldots,10^{-2}\}$, batch sizes from 64 to 1000, and environment sample counts $S$ from 3 to 50. We search the structural weight $\lambda$, gate weight $\alpha$, and dropout ratio in $\{0,0.1,\ldots,1.0\}$. Equation~\eqref{loss} uses $\beta$ for weight decay. The polynomial degree is fixed at $p=2$ following~\citet{liu2024identifiable}. We tune the environment dimension $K$ from 1 to 20 and the $\boldsymbol{x}_2$ embedding dimension from 100 to 500. Validation selects the conditional-environment family from Gaussian, Gamma, Beta, Dirichlet, and Exponential distributions. Poisson, Bernoulli, Multinomial, and Negative Binomial appear only in the sensitivity analysis outside Theorem~\ref{thm:ident}.

\section{Synthetic generators}\label{app:synthetic_generators}
\paragraph{Structural mechanisms.} For the linear setting:
\begin{equation*}
    \begin{aligned}
        &\varepsilon_i :\sim
            \begin{cases}
                \mathcal{B}(\alpha, \beta), & \text{if } \varepsilon \sim \text{Beta} \\
                \mathcal{G}(\alpha, \beta), & \text{if } \varepsilon \sim \text{Gamma}
            \end{cases} \\
        &z_{2,1} := \varepsilon_1, \\
        &z_{2,2} := \lambda_{1,2}(e)z_{2,1} + \varepsilon_2, \\
        &z_{2,3} := \lambda_{2,3}(e)z_{2,2} + \varepsilon_3, \\
        &z_{2,4} := \lambda_{2,4}(e)z_{2,2} + \varepsilon_4, \\
        &z_{2,5} := \lambda_{1,5}(e)z_{2,1} + \varepsilon_5, \\
        &z_{2,6} := \lambda_{4,6}(e)z_{2,4} + \varepsilon_6, \\
        &z_{2,7} := \lambda_{5,7}(e)z_{2,5} + \varepsilon_7,
    \end{aligned}
\end{equation*}
where distribution parameters $\alpha, \beta \in \mathcal{U}[0.1, 2.0]$ and coefficients $\lambda_{i,j}(e) \in [-1.0,-0.5] \cup [0.5, 1.0]$. For the nonlinear setting with Gaussian noise:
\begin{equation*}
        \begin{aligned}
        &\varepsilon_i \sim \mathcal{N}(\alpha, \beta), \\
        &z_{2,1} := \varepsilon_1, \\
        & z_{2,2} := \lambda_{1,2}(e)z_{2,1}^2 + \varepsilon_2, \\
        &z_{2,3} := \lambda_{2,3}(e)z_{2,2} + \varepsilon_3, \\
        &z_{2,4} := \lambda_{2,4}(e)z_{2,2}z_{2,3} + \varepsilon_4, \\
        &z_{2,5} := \lambda_{1,5}(e)z_{2,1}^2 + \varepsilon_5, \\
        &z_{2,6} := \lambda_{4,6}(e)z_{2,4} + \varepsilon_6, \\
        &z_{2,7} := \lambda_{5,7}(e)z_{2,5}z_{2,3} + \varepsilon_7.
        \end{aligned}
\end{equation*}

\paragraph{Paired comparisons and uncertainty.} The controlled comparisons use \(S=10\) paired seeds. Let \(m_s^{(1)}\) and \(m_s^{(2)}\) denote the metric values under the first and second conditions for seed \(s\). We orient each paired difference so that a positive value favors the first condition:\begin{equation}d_s=\begin{cases}m_s^{(1)}-m_s^{(2)}, & \text{higher-is-better metric},\\m_s^{(2)}-m_s^{(1)}, & \text{lower-is-better metric}.\end{cases}\end{equation}The condition means and reported improvement are\begin{equation}\bar m^{(j)}=\frac{1}{S}\sum_{s=1}^{S}m_s^{(j)},\quad j\in\{1,2\},\qquad\widehat{\Delta}=\frac{1}{S}\sum_{s=1}^{S}d_s.\end{equation}Thus, before rounding, the improvement equals the oriented difference between the two condition means. Its uncertainty depends on the seed-wise paired differences. For bootstrap replicate \(b\), we sample \(S\) indices with replacement and compute\begin{equation}\widehat{\Delta}^{*(b)}=\frac{1}{S}\sum_{r=1}^{S}d_{I_r^{(b)}},\qquad I_r^{(b)}\overset{\mathrm{iid}}{\sim}\mathrm{Unif}\{1,\ldots,S\}.\end{equation}With \(B=10{,}000\) replicates, the percentile-bootstrap interval is\begin{equation}\mathrm{CI}_{0.95}=\left[Q_{0.025}\!\left(\widehat{\Delta}^{*}\right),Q_{0.975}\!\left(\widehat{\Delta}^{*}\right)\right],\end{equation}where \(Q_p\) is the empirical \(p\)-quantile of the \(B\) bootstrap means.\par\paragraph{Density misspecification.} Density evaluation draws user-conditioned variables from Laplace, Gamma, Beta, and Negative Binomial families and uses a Gaussian model as the misspecified reference.
\section{Auxiliary gate analysis}\label{app:gate_analysis}

We evaluate the gate across multilayer perceptron (MLP) dimensions to test its architectural stability. Fig.~\ref{fig:gate_stability} compares CILER with and without the gate on every dataset and metric. The paired-test $p$-values for Recall, NDCG, and MRR are $0.885$, $0.573$, and $0.105$ on Synthetic. They are $0.594$, $0.547$, and $0.444$ on Meituan and $0.591$, $0.104$, and $0.011$ on Yelp. Only Yelp MRR is significant at the $0.05$ level. The gate acts as a dataset- and metric-dependent decoder regularizer.

\begin{figure}
\centering
\includegraphics[width=1.0\columnwidth]{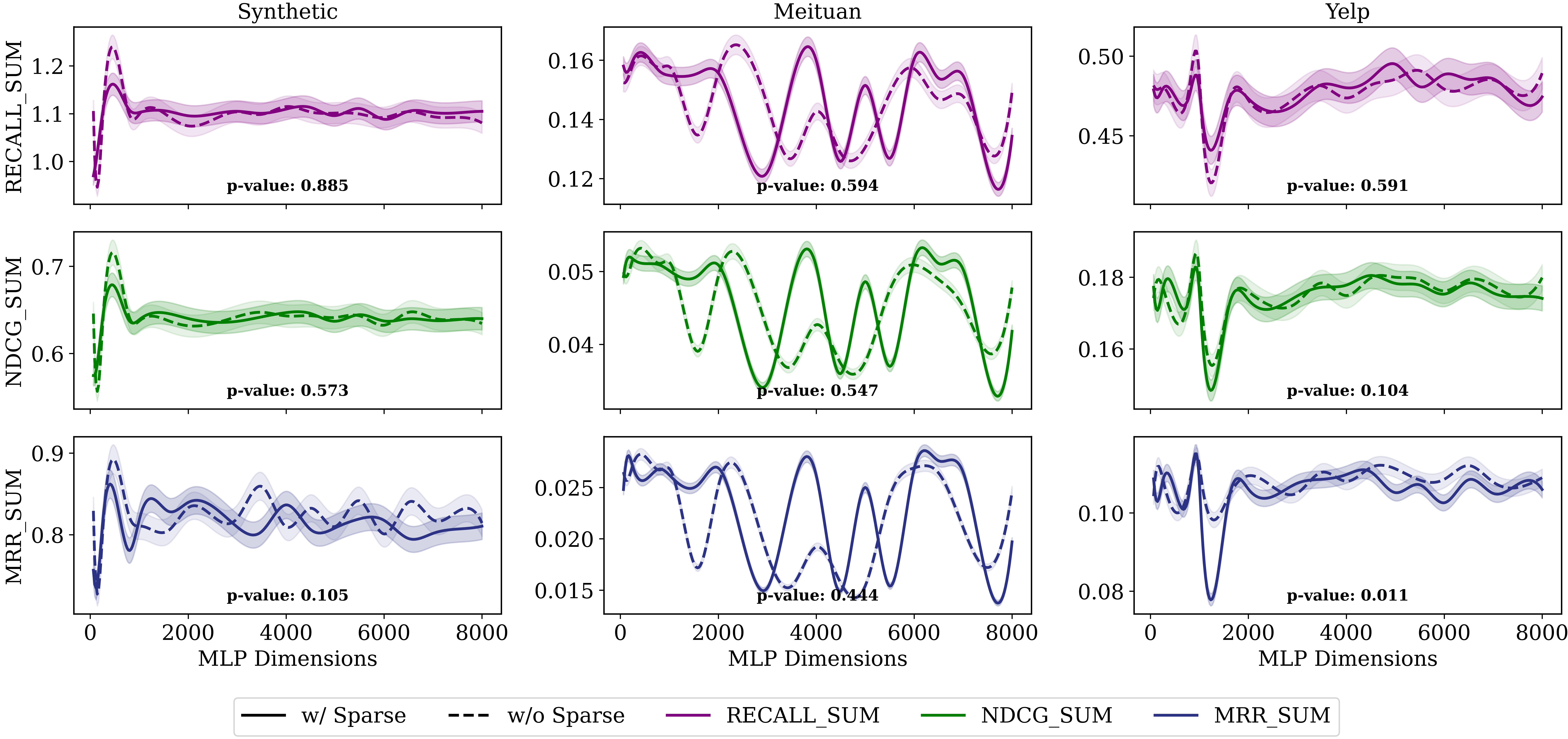}
\caption{Effect of the auxiliary gate across MLP dimensions. Solid lines denote CILER with the gate. Dashed lines denote the model without the gate. Columns report Recall, NDCG, and MRR, and rows correspond to Synthetic, Meituan, and Yelp.}
\label{fig:gate_stability}
\end{figure}

\section{Causal adjustment diagnostic}\label{app:tce}

Proposition~\ref{prop:risk} bounds excess deployment risk at the level of the objective. This appendix inspects the same mechanism for individual users. Adjustment replaces the encoded environment-sensitive block $\boldsymbol{z}_2$ with its structural reconstruction $\boldsymbol{z}_2^{\prime}$ of Eq.~\eqref{causalZ2}, and the total causal effect (TCE) records how far that replacement moves the item logits~\citep{pearl2003causality},
\begin{equation}
\mathrm{TCE}=\mathbb{E}\!\left[f_{\theta_3}(\boldsymbol{W}_z\odot[\boldsymbol{z}_1,\boldsymbol{z}_2^{\prime}])\right]-\mathbb{E}\!\left[f_{\theta_3}(\boldsymbol{W}_z\odot[\boldsymbol{z}_1,\boldsymbol{z}_2])\right].
\end{equation}
The per-user quantity averages the same difference over the $I$ items of the catalogue,
\begin{equation}
\mathrm{TCE}_u=\frac{1}{I}\sum_{i=1}^{I}\left(f_{\theta_3}(\boldsymbol{W}_z\odot[\boldsymbol{z}_1,\boldsymbol{z}_2^{\prime}])_i-f_{\theta_3}(\boldsymbol{W}_z\odot[\boldsymbol{z}_1,\boldsymbol{z}_2])_i\right).
\end{equation}

\begin{figure}
\centering
\includegraphics[width=1.0\columnwidth]{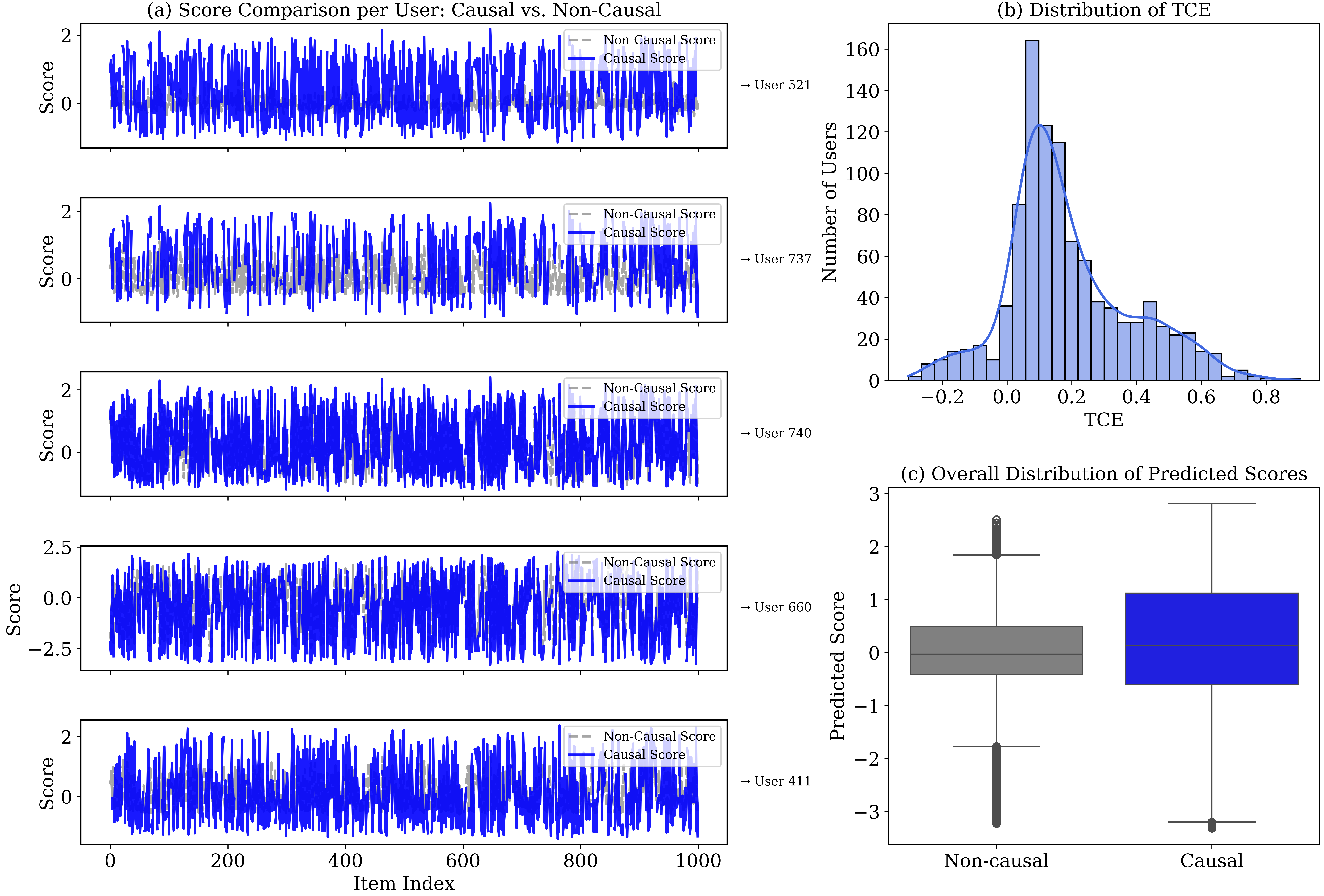}
\caption{Predicted scores before and after causal adjustment on Synthetic. Panel (a) compares the adjusted and unadjusted score sequences for five users across the item catalogue. Panel (b) reports the distribution of the per-user total causal effect. Panel (c) compares the two populations of predicted scores.}
\label{fig:tce_adjustment}
\end{figure}
Fig.~\ref{fig:tce_adjustment} evaluates the diagnostic on Synthetic. Panel (a) traces the two score sequences for five users across the catalogue. The sequences follow each other closely, so the replacement preserves the item-level shape of the predictions. Panel (b) shows $\mathrm{TCE}_u$ concentrated just above zero with a right tail, which places the effect of the structural block at a modest positive shift for most users. Panel (c) shows a wider interquartile range and fewer extreme outliers after the replacement. The three panels put the contribution of the environment-sensitive block in a bounded range. The diagnostic remains descriptive and does not estimate the constant in Proposition~\ref{prop:risk}.

\FloatBarrier
\bibliographystyle{elsarticle-num-names}
\bibliography{refs}

\end{document}